\ifx\documentclass\undefined
\documentstyle[12pt]{article}
\else
\documentclass[12pt]{article}
\fi

\usepackage{color}

\usepackage{ascmac}

\newcommand{\beq}{\begin{eqnarray*}}
\newcommand{\eeq}{\end{eqnarray*}}
\makeatletter
\renewcommand{\theequation}{\thesection.\arabic{equation}}
\@addtoreset{equation}{section}
\def\eqnarray{%
\stepcounter{equation}%
\let\@currentlabel=\theequation
\global\@eqnswtrue
\global\@eqcnt\z@
\tabskip\@centering
\let\\=\@eqncr
$$\halign to \displaywidth\bgroup\@eqnsel\hskip\@centering
$\displaystyle\tabskip\z@{##}$&\global\@eqcnt\@ne
\hfil$\displaystyle{{}##{}}$\hfil
&\global\@eqcnt\tw@$\displaystyle\tabskip\z@{##}$\hfil
\tabskip\@centering&\llap{##}\tabskip\z@\cr}
\makeatother
\newtheorem{theorem}{Theorem}[section]
\newtheorem{lemma}[theorem]{Lemma}
\newtheorem{corollary}[theorem]{Corollary}
\newtheorem{proposition}[theorem]{Proposition}
\newtheorem{remark}{Remark}[section]

\newsavebox{\toy}
\savebox{\toy}{\framebox[0.65em]{\rule{0cm}{1ex}}}
\newcommand{\QED}{\usebox{\toy}}
\def\nlni{\par\ifvmode\removelastskip\fi\vskip\baselineskip\noindent}
\newenvironment{proof}{\nlni\begingroup\it Proof.\rm}{
\endgroup\vskip\baselineskip}

\begin{document}
\setlength{\baselineskip}{15pt}
\title{
Level statistics for one-dimensional Schr\"odinger operators and Gaussian beta ensemble\\ 
}
\author{
Fumihiko Nakano
\thanks{
Department of Mathematics,
Gakushuin University,
1-5-1, Mejiro, Toshima-ku, Tokyo, 171-8588, Japan.
e-mail : 
fumihiko@math.gakushuin.ac.jp}
}
\maketitle
\begin{abstract}
We study
the level statistics for two classes of 1-dimensional random Schr\"odinger operators : 
(1)
for operators whose coupling constants decay as the system size becomes large, 
and 
(2)
for operators with critically decaying random potential. 
As a byproduct of (2) 
with our previous result \cite{KN} 
imply the coincidence of the limits of 
circular and Gaussian beta ensembles.
\end{abstract}


\section{Introduction}
As one of the recent developments 
of the theory of random matrices, 
the continuum limit of the beta ensembles are recently revealed : 
Killip-Stoiciu \cite{KS}
identified the limit of the circular beta ensemble($C_{\beta}$-ensemble, in short) by using the solution to a SDE.
Valk\'o-Vir\'ag \cite{VV}
identified the limit of Gaussian beta ensemble($G_{\beta}$-ensemble, in short) by using Brownian carousel. 
At the same time, 
it also gave a new insight to 
the level statistics problem of 1-dimensional random Schr\"odinger operators : 
In 
\cite{KS}, 
they also studied the level statistics problem of the CMV matrices, 
that is, 
they studied the scaling limit of the point process 
$\xi_L$
whose atoms are composed of the scaled eigenvalues of the truncated matrices. 
When the  diagonal components 
decay in the order of 
$n^{-\alpha}$ 
they showed that, 
$\xi_L$
converges to 
(i)
$\alpha > \frac 12$ : 
the clock process, 
(ii)
$\alpha < \frac 12$ : 
the Poisson process,
(iii)
$\alpha = \frac 12$ : 
the limit of 
$C_{\beta}$-ensemble. 
In \cite{KVV}, 
they studied the same problem for 1-dimensional discrete Schr\"odinger operators whose random potential decays in the order of 
$n^{- \frac 12}$
and showed that 
$\xi_L$
converges to the limit of 
$G_{\beta}$-ensemble. 
Moreover, 
they also studied the random Hamiltonians with system size 
$L$
in which the coupling constant decays in the order of 
$L^{- \frac 12}$. 
They identified the limit(``Sch$_{\tau}$") of 
$\xi_L$
and studied its various properties.
In \cite{KN}, 
they studied the 1-dimensional Schr\"odinger operators in the continuum with random decaying potential, for the case of 
$\alpha > \frac 12$ 
and 
$\alpha = \frac 12$, 
and the results obtained are parallel to that in \cite{KS}.
This paper 
is basically a continuum analogue of \cite{KVV} : 
(1)
we consider the operator on 
$[0,L]$ 
where the coupling constant is equal to 
$L^{-\alpha}$. 
We study the limit of 
$\xi_L$ 
for the case of 
$\alpha > \frac 12$
and 
$\alpha = \frac 12$. 
(2)
we consider the same operator to that in \cite{KN} 
for the critical decay
$\alpha = \frac 12$
and show that 
$\xi_L$ 
converges to the limit of $G_{\beta}$-ensemble, 
which, together with the results in \cite{KN}, 
implies that the limit of these two beta ensembles are equal. 
In the next subsection, 
we shall explain the motivation of the problem (1)\footnote{The author 
would like to thank F. Klopp for introducing this problem. }. 
\subsection{Motivation and Set ups}
The localization length 
$l_{loc}$
of the 1-dimensional Schr\"dingier operator 
$H = - \triangle + \lambda V$
is typically in the order of 
$\lambda^{-2}$. 
Thus, setting 
$H_L := H |_{[0, L]}$, 
$\lambda = L^{- \alpha}$, 
we expect : 
\begin{enumerate}

\item[(1)]
(extended case)
$\alpha > \frac 12$ : 
we have 
$L \ll l_{loc}$
so that the particle would be extended. 

\item[(2)]
(localized case)
$\alpha < \frac 12$ : 
we have
$l_{loc} \ll L$
so that the particle would be localized. 

\item[(3)]
(critical case)
$\alpha = \frac 12$ : 
$l_{loc} \simeq L$
so that it would correspond to the critical case.

\end{enumerate}

Therefore
if we consider the level statistics problem, 
$\xi_L$
would converge to 
(1)
$\alpha > \frac 12$ : 
the clock process,
(2)
$\alpha < \frac 12$ : 
the Poisson process, 
(3)
$\alpha = \frac 12$ : 
something which is intermediate between the clock and Poisson.

In 
this paper we consider this problem in the continuum setting : 
The Hamiltonian is defined by 
\beq
H_L &:=& H_{\lambda_L} |_{[0, L]}
\eeq
with Dirichlet boundary condition, 
where 
$H_{\lambda_L}$
is the Schr\"odinger operator with the coupling constant 
$\lambda_L$
which decays at certain rate as the system size
$L$
is large : 
\beq
H_{\lambda_L}
&:=& -\frac {d^2}{dt^2} + \lambda_L F(X_t), 
\quad
\lambda_L := L^{- \alpha}, 
\quad
\alpha > 0. 
\eeq
$(X_t)_{t \ge 0}$
is a Brownian motion on a compact Riemannian manifold
$M$
and 
$F \in C^{\infty}(M)$
with 
\[
\langle F \rangle := \int_M F(x) dx = 0.
\]
Let 
$\{ E_n (L) \}_{n \ge 1}$
be the eigenvalues of 
$H_L$
in the increasing order. 
Since 
we only consider the positive eigenvalues, we set 
\[
n(L) := \min \{ n | E_n (L) > 0 \}.
\]
Fix the reference energy 
$E_0 > 0$
arbitrary.
To study the local distribution of 
$E_n(L)$'s
near 
$E_0$, 
we set 
\begin{eqnarray}
\xi_L &:=&
\sum_{n \ge n(L)}
\delta_{L (\sqrt{E_n(L)} - \sqrt{E_0})}.
\label{Pointprocess}
\end{eqnarray}
Here we take 
$\sqrt{ E_n (L) }$
instead of 
$E_n (L)$
to unfold the eigenvalues with respect to the density of states. 
Our purpose is to study the behavior of 
$\xi_L$
as 
$L$
tends to infinity.
%
%
\subsection{Results for Extended Case}
If we consider the free Laplacian, 
we must take a subsequence 
in order that 
$\xi_L$
converges to a point process.
We need the same condition described below. \\

{\bf (A)}
The subsequence 
$\{ L_j \}_{j=1}^{\infty}$
satisfies 
$L_j \stackrel{j \to \infty}{\to} \infty$
and
\[
\sqrt{E_0} L_j
=
m_j \pi + \beta + o(1), 
\quad
j \to \infty
\]
$m_j \in {\bf N}$, 
$\beta \in [0, \pi)$.
\begin{theorem}
\label{clock}
Assume 
$(A)$
and 
$\alpha > \frac 12$. 
Then we have
\beq
\lim_{j \to \infty}
{\bf E}[ e^{- \xi_{L_j}(f)} ]
=
\exp \left(
- \sum_{n \in {\bf Z}}
f 
(n \pi - \beta)
\right).
\eeq
In other words, 
$\xi_{L_j}$
converges to a (deterministic) clock process with spacing 
$\pi$, 
in probability.
\end{theorem}
When 
the random potential is spatially decaying in the order of 
$\alpha > \frac 12$, 
$\xi_L$
also converges to a clock process but 
$\beta$
is random \cite{KN}.
Here 
the effect of the random potential is rather weak compared to that in \cite{KN}. 
In fact, 
the solution to the eigenvalue equation 
$H x_t = E x_t$
approaches to the free solution in probability
(Theorem \ref{free}).
However, 
the randomness appear in the second order
(Theorem \ref{2ndLimitTheorem}).
To see the spacing between eigenvalues, 
we renumber the eigenvalues near 
$E_0$
such that 
$
\cdots < E'_{-2} < E'_{-1} < E_0 \le E'_0 < E'_1 < \cdots
$.
Then by the argument of the proof of Theorem \ref{clock}, for any 
$l \in {\bf Z}$, 
\[
\left( 
\sqrt{E'_{l+1}(L)} - \sqrt{E'_l(L)} 
\right) 
L - \pi
\to 0
\]
in probability.
Hence it is reasonable to consider 
\beq
X_j(n) 
:=
\left\{
\left(
\sqrt{E_{m_j+n+1}(L_j)} - \sqrt{E_{m_j+n+j}(L_j)}
\right)L_j - \pi
\right\}
L_j^{\alpha - \frac 12}, 
\quad
n \in {\bf Z}
\eeq
to study the second order asymptotics of eigenvalues near 
$E_0$, 
for 
$E_{m_j}(L_j)$
may be regarded as the closest eigenvalue to 
$E_0$. 
\begin{theorem}
\label{2ndLimitTheorem}
$\{ X_j (n) \}_{n \in {\bf Z}}$
converges in distribution to the Gaussian system with covariance 
\beq
C(n, n')
&:=&
\frac {C(E_0)}{8 E_0}
\cases{
2 & $(n=n')$ \cr
-1 & $(|n - n'| =1)$ \cr
0 & $(otherwise)$ \cr
}
\eeq
where
\beq
C(E)
&:=&
\int_M 
| \nabla (L + 2i \sqrt{E})^{-1} F |^2 dx
\eeq
and 
$L$
is the generator of 
$(X_t)$. 
\end{theorem}
Here 
the covariance is short range, while it is not the case if the potential is spatially decaying \cite{KN}. 
\begin{remark}
By definition of 
$X_j(n)$, 
we have
\beq
\sqrt{E_{m_j + n}(L_j)}
&=&
\sqrt{E_{m_j}(L_j)}
+
\frac {n \pi}{L_j}
+
\frac {1}{L_j^{\alpha + \frac 12}}
\sum_{l=0}^{n-1}
X_j (l)
\eeq
so that the difference between 
$\sqrt{E_{m_j+n}(L_j)}$
and 
$\sqrt{E_{m_j}(L_j)}$
converges to the Gaussian in the second order. 
\end{remark}
\begin{remark}
Suppose that 
we consider two reference energies 
$E_0, E'_0$, 
$E_0 \ne E'_0$
both satisfying 
(A)
with the same subsequence.
Then the corresponding 
$\{ X_j (n) \}$, $\{ X'_j(n') \}$
converge jointly to the two independent Gaussian systems each other.
The same property
also holds for 
Theorem \ref{Schtau}, \ref{Carousel} 
stated below.
\end{remark}

%
\subsection{Results for Critical Case}
In this subsection we set 
$\alpha = \frac 12$. 
By
Theorem \ref{Delocalized}, 
the solution to the eigenvalue equation 
$H x_t = E x_t$
is bounded from above and below so that it is different from that in the critically decaying potential case studied in \cite{KN}. 
\begin{theorem}
\label{Schtau}
Assume 
$\alpha = \frac 12$
and 
(A). 
Then we have 
\beq
\lim_{j \to \infty}
{\bf E}[ e^{- \xi_{L_j}(f)} ]
=
{\bf E}\left[
\exp \left(
- \sum_{n \in {\bf Z}}
f (\Psi_1^{-1}(2 n \pi - 2 \beta))
\right)
\right].
\eeq
where 
$\Psi_t (c)$
is a strictly-increasing function valued process such that for any 
$c_1, c_2, \cdots, c_m$, 
$\Psi_t(c_1), \cdots, \Psi_t(c_m)$
jointly satisfy the following SDE. 
\begin{eqnarray}
d \Psi_t(c_j)
&=&
\left(
2c_j - Re \; \frac {i}{2 E_0} \langle F g_{\sqrt{E_0}}\rangle
\right) dt
\nonumber
\\
&& \qquad + 
\frac {1}{\sqrt{E_0}}
\Biggl\{
\sqrt{
\frac {C(E_0)}{2}
}
Re \; 
\left(
e^{i \Psi_t(c_j)} d Z_t
\right)
+
\sqrt{ C(0) }
dB_t
\Biggr\}
\label{SchtauSDE}
\end{eqnarray}
$j=1, 2, \cdots, m$, 
where
$Z_t$
is a complex Brownian motion independent of a Brownian motion 
$B_t$
and 
\beq
g_{\sqrt{E_0}}
&:=&
(L + 2i \sqrt{E_0})^{-1} F, 
\quad
g := L^{-1} (F - \langle F \rangle),
\\
C(E_0)
&:=&
\int_M | \nabla g_{\sqrt{E_0}} |^2 dx, 
\quad
C(0)
:=
\int_M | \nabla g_{} |^2 dx.
\eeq
\end{theorem}
This 
SDE
is the same as that satisfied by the phase function of 
``Sch$_{\tau}$" \cite{KVV} 
up to constant. 
Hence 
the properties of 
``Sch$_{\tau}$" 
derived in \cite{KVV} 
such as 
strong repulsion, 
large gap asymptotics, 
explicit form of intensity and 
CLT, 
also hold for our case. 
%
%
\subsection{Results for decaying potential model with critical decay}
%
%
In this subsection we consider 
\beq
H = - \frac {d^2}{d t^2} + V(t), 
\quad
V(t) := a(t) F(X_t).
\eeq
where 
$a \in C^{\infty}$, 
$a(-t) = a(t)$, 
$a$
is decreasing on 
$[0, \infty)$, 
and 
\[
a(t) = t^{-\frac 12}(1 + o(1)), 
\quad
t \to \infty.
\]
$(X_t)$
and 
$F$
satisfy the same conditions stated in subsection 1.1.
Let 
$H_L := H |_{[0, L]}$
be the finite box Hamiltonian with Dirichlet boundary condition and let 
$\{ E_n (L)\}_{n \ge n(L)}$
be the set of positive eigenvalues of 
$H_L$. 
Let 
$\xi_L$
defined as in 
(\ref{Pointprocess}).
In \cite{KN}, 
we proved that 
$\xi_L$
converges to the limit of $C_{\beta}$-ensemble.
That is, let 
$\zeta^C_{\beta} = \sum_k \delta_{\lambda_k}$
be the continuum limit of the $C_{\beta}$-ensemble.
Then 
\begin{equation}
\xi_L \stackrel{d}{\to} \tilde{\zeta}^C_{\beta}
\label{Circular}
\end{equation}
where
$\tilde{\zeta}^C_{\beta}
=
\sum_k \delta_{\lambda_k /2}$
and 
$\beta = \frac {8 E_0}{C(E_0)}$. 
Here we give
a different description of the limit.
\begin{theorem}
\label{Carousel}
The limit 
$\xi_{\infty} = \lim_{L \to \infty} \xi_L$
has the following property.
Let 
$N(\lambda) := \sharp \{ 
\mbox {atoms of $\xi_{\infty}$ in } [0, \lambda] \}$
be the counting function of 
$\xi_{\infty}$. 
Then 
$N(\lambda) \stackrel{d}{=} \frac {1}{2 \pi}\Psi_{1-}(\lambda)$
where
$\Psi_t (\lambda)$, $t \in [0, 1)$
is the strictly-increasing function valued process which is the solution to 
\begin{eqnarray}
d \Psi_t (\lambda) &=& 2 \lambda dt 
+ 
\frac {D(E_0)}{\sqrt{1-t}}
Re \left[ 
(e^{i \Psi_t(\lambda)} -1) d Z_t
\right], 
\quad
\Psi_0(\lambda) = 0
\label{SDEeq}
\end{eqnarray}
where
$
D(E_0) = \sqrt{
\frac {C(E_0)}{2 E_0}
}
$.
\end{theorem}
This theorem 
is the continuum analogue of that in 
\cite{KVV}. 
To see the significance of Theorem \ref{Carousel}, 
let us recall the Gaussian beta ensemble
whose joint density of ordered eigenvalues
$\lambda_1 \le \lambda_2 \le \cdots \le \lambda_n$ 
is proportional to 
\[
\exp \left(
- \frac {\beta}{4} \sum_{k=1}^n \lambda_k^2
\right)
\prod_{j < k}
| \lambda_j - \lambda_k |^{\beta}.
\]
Then 
Valk\'o - Vir\'ag
found the following representation of the continuum limit of the G$_{\beta}$- ensemble.
\begin{theorem}
{\bf (Valk\'o - Vir\'ag \cite{VV})}
\\
Let 
$\mu_n$
be the sequence such that 
$n^{\frac 16} (2 \sqrt{n} - | \mu_n |) \to \infty$. 
Then
\beq
\sum_{k}
\delta_{\Lambda_k^{(n)}}
\stackrel{d}{\to} 
\zeta^G_{\beta}, 
\quad
\Lambda_k^{(n)}
:= \sqrt{4n - \mu_n^2} (\lambda_k - \mu_n)
\eeq
where
$\sharp \{ \mbox{atoms of $\zeta_{\beta}^G$ in } [0,\lambda] \}
\stackrel{d}{=}
\frac {1}{2 \pi}
\alpha_{\infty}(\lambda)$
where
$\alpha_{\infty}(\lambda) := \lim_{t \to \infty} \alpha_t(\lambda)$
and 
$\alpha_t(\lambda)$
is the solution to 
\begin{equation}
d \alpha_t(\lambda)
=
\lambda 
\frac {\beta}{4} e^{- \frac {\beta t}{4}}dt
+
Re [ (e^{ i\alpha_t(\lambda) } - 1) d Z_t ], 
\quad
\alpha_0 (\lambda) = 0.
\label{Sinebeta}
\end{equation}
\end{theorem}
Note that 
$\alpha_{\infty}(\lambda) \in 2 \pi {\bf Z}$.
By the time change  
$t = 1 - e^{-cs}$, 
$c = \frac {\beta}{4}$, 
$\beta = \frac {8 E_0}{C(E_0)}$, 
SDE (\ref{SDEeq})
is transformed to 
(\ref{Sinebeta})\cite{KVV}. 
Therefore 
\begin{corollary}
\label{ConvergetoGaussianbeta}
For 
$\zeta^G_{\beta} = \sum_k \delta_{\lambda_k}$, 
let 
$\tilde{\zeta}^G_{\beta} = \sum_k \delta_{\lambda_k/2}$.
Then
$
\xi_{L} \stackrel{d}{\to} \tilde{\zeta}^G_{\beta}
$
where
$\beta = \frac {8 E_0}{C(E_0)}$. 
\end{corollary}
By varying 
$E_0 \in (0, \infty)$
or 
$F$, 
any 
$\beta > 0$
can be realized.
Hence 
by combining (\ref{Circular}) and Corollary \ref{ConvergetoGaussianbeta}, 
we have the coincidence of the limit of two beta ensembles. 
\begin{corollary}
For any 
$\beta > 0$, 
$\zeta_{\beta}^C = \zeta_{\beta}^G$.
\end{corollary}
This is known for 
$\beta = 1, 2, 4$. 
Valk\'o-Vir\'ag
have a direct proof of this fact by showing that these two descriptions are equivalent\cite{V}. 

The method of proof of 
Theorem \ref{clock} - \ref{Schtau}
is essentially the same as that in \cite{KN} : 
we write the Laplace transform of 
$\xi_L$
in terms of the Pr\"ufer variables, and study the behavior of the relative Pr\"ufer phase. 
The idea of proof of 
Theorem \ref{Carousel}
is the same as that in 
\cite{VV, KVV}
with different techniques : 
we identify the scaling limit of the relative Pr\"ufer phase as the solution to a SDE, and show that  
$t \uparrow 1$ 
limit of this solution gives the counting function of 
$\xi_{\infty}$. 
The outline of this paper is as follows : 
Section 2 is the preparation of the basic notations and tools.
In Section 3 - 5, 
we prove 
Theorem \ref{clock}, \ref{2ndLimitTheorem}, 
Theorem \ref{Schtau}, 
and 
Theorem \ref{Carousel}
respectively. 
In 
Appendix, 
we recall the techniques used in \cite{KU, KN}. 
In what follows, 
$C$
denotes positive constants which is subject to change from line to line.
%
%
\section{Preparation}
For general 1-dim Schr\"odinger operator 
$H = - \frac {d^2}{dt^2} + q$, 
let 
$x_t$
be the solution to the equation 
$H x_t = \kappa^2 x_t$, $x_0=0$, 
$(\kappa > 0)$
which we write in the (modified) Pr\"ufer variables :  
\begin{equation}
\left( \begin{array}{c}
x_t \\ x'_t /\kappa
\end{array} \right)
=
r_t
\left( \begin{array}{c}
\sin \theta_t \\ \cos \theta_t
\end{array} \right), 
\quad
\theta_0 = 0.
\label{Prufer}
\end{equation}
We define 
$\tilde{\theta}_t(\kappa)$
by 
\[
\theta_t (\kappa) = \kappa t + \tilde{\theta}_t (\kappa).
\]
Then it follows that 
\begin{eqnarray}
r_t(\kappa)
&=&
\exp \left(
\frac {1}{2\kappa} Im 
\int_0^t 
q(s) e^{2i \theta_s(\kappa)} ds
\right)
\label{r-eq}
\\
\tilde{\theta}_t (\kappa)
&=&
\frac {1}{2 \kappa}
\int_0^t
Re (e^{2i \theta_s(\kappa)} -1 ) q(s)ds, 
\quad
\tilde{\theta}_0(\kappa) = 0
\label{theta-eq}
\\
\frac {\partial \theta_t(\kappa)}{\partial \kappa}
&=&
\int_0^t 
\frac {r_s^2}{r_t^2} ds
+
\frac {1}{2 \kappa^2}
\int_0^t 
\frac {r_s^2}{r_t^2}
q(s)
(1 - Re \; e^{2i \theta_s(\kappa)}) ds.
\label{theta-kappa}
\end{eqnarray}
Since 
$r_t$
is bounded from below (Theorem \ref{free}, \ref{Delocalized} and \cite{KU} Proposition 2.1), 
for any closed interval 
$I \subset (0, \infty)$
we have 
$\inf_{\kappa \in I}\frac {\partial \theta_t(\kappa)}{\partial \kappa} > 0$
for sufficiently large 
$t > 0$
so that  
$\theta_t(\kappa)$
is strictly-increasing as a function of 
$\kappa \in I$. 
Set 
$q(t) := \lambda_L F(X_t)$
and let 
$\theta_{t, L}(\kappa)$, 
$r_{t, L}(\kappa)$
be the corresponding Pr\"ufer variables.
Let 
\beq
\kappa_0 &:=& \sqrt{E_0}
\\
\Psi_{t,L}(x)
&:=&
\theta_{t,L} \left(
\kappa_0 + \frac xL
\right) - \theta_{t,L}(\kappa_0), 
\quad
0 \le t \le L
\eeq
be the relative Pr\"ufer phase.
Further write 
$\theta_{L, L}(\kappa_0)$
as 
\beq
\theta_{L, L} (\kappa_0) 
&=&
m(\kappa_0, L) \pi + \phi(\kappa_0, L), 
\quad
m(\kappa_0, L) \in {\bf N}, 
\quad
\phi(\kappa_0, L) \in [0,  \pi).
\eeq
Then 
we have the following representation of the Laplace transform of 
$\xi_L$
in terms of Pr\"ufer variables \cite{KS}.  
\begin{lemma}
\label{Laplace}
\beq
{\bf E}[ e^{- \xi_L (f)} ]
&=&
{\bf E} \left[
\exp \left(
- \sum_{n \ge n(L) - m(\kappa_0, L)}
f \left(
( \Psi_{L, L}^{-1})
( n \pi -  \phi(\kappa_0, L))
\right)
\right)
\right]
\eeq
for 
$f \in C_c^+({\bf R})$, 
where
$\xi_L(f) := \int_{\bf R} f d \xi_L$. 
\end{lemma}
\begin{proof}
Let 
\[
x_n = L \left(
\sqrt{E_n(L)} - \kappa_0
\right)
\]
be the $n$-th atom of 
$\xi_L$. 
By Sturm oscillation theorem, 
$x = x_n$
if and only if 
$\theta_{L, L}(\kappa_0+\frac xL) = n \pi$ 
so that we have 
\beq
\{ x_n \}_{n \ge n(L)}
&=&
\left\{ x 
\, | \,
 \Psi_{L, L}(x) =  n \pi - \phi(\kappa_0, L) , 
\quad
n \ge n(L) - m(\kappa_0, L)
\right\}.
\eeq
\QED
\end{proof}
%
\section{Extended case}
Throughout this section we set 
$\alpha > \frac 12$
to study the extended case. 
\subsection{Proof of Theorem \ref{clock}}
First of all, we study the behavior of the following quantity. 
\begin{equation}
J_{t, L}(\kappa)
:=
\lambda_L
\int_0^t 
e^{2i \theta_{s,L}(\kappa)}  F(X_s) ds, 
\quad
\kappa \ge 0
\label{Jdefine}
\end{equation}
where we set 
$\theta_{t, L}(0) := 0$
for convenience. 
\begin{lemma}
\label{J1}
Let 
$\kappa \ge 0$.
\\
(1)
\[
\sup_{0 \le t \le L} | J_{t, L}(\kappa) | 
\stackrel{L \to \infty}{\to} 
0
\]
in probability, compact uniformly for 
$\kappa \in (0, \infty)$. 
\\
(2)
\begin{eqnarray}
L^{\alpha - \frac 12}J_{t, L}(\kappa)
&=&
L^{\alpha - \frac 12} Y_{t, L}(\kappa) + o(1)
\nonumber
\end{eqnarray}
as 
$L \to \infty$, 
where 
\begin{eqnarray}
Y_{t, L} (\kappa)
&:=&
\lambda_L
\int_0^t 
e^{2i \theta_{s, L}(\kappa)} d M_s (\kappa)
\label{Y}
\end{eqnarray}
and 
$M_s(\kappa)$
is a complex martingale defined in Lemma \ref{PartialIntegration}.
\end{lemma}
\begin{proof}
(1)
We first assume that 
$\kappa > 0$.
By
Lemma \ref{PartialIntegration}
\beq
J_{t, L}(\kappa)
&=&
\lambda_L 
\Biggl\{
\left[
e^{2i \theta_{s,L}(\kappa)} g_{\kappa} (X_s)
\right]_0^t
\\
&& - 
\frac {2i}{2 \kappa}
\int_0^t Re \; 
(e^{2i \theta_{s,L}(\kappa)} - 1)
\lambda_L F(X_s) 
e^{2i \theta_{s,L}(\kappa)} 
g_{\kappa}(X_s) ds
\\
&& + 
\int_0^t 
e^{2i \theta_{s,L}(\kappa)}
d M_s(\kappa)
\Biggr\}
\\
&=:& I_{t, L} + II_{t, L} + Y_{t, L}.
\eeq
Then
$I_{t, L} = O(\lambda_L) = O(L^{-\alpha})$, 
$II_{t, L} \le
C \lambda_L^2 t
=
O(L^{-2 \alpha + 1})$
and 
$Y_{t, L}(\kappa)$
is a martingale satisfying 
\[
\langle Y_{\cdot,L}(\kappa), Y_{\cdot,L}(\kappa) \rangle_t, 
\quad
\langle Y_{\cdot,L}(\kappa), \overline{Y_{\cdot,L}(\kappa)} \rangle_t
\le 
C \lambda_L^2 t.
\]
Therefore by martingale inequality we have 
\beq
{\bf E}[ \sup_{0 \le t \le L} | Y_{t, L}(\kappa) |^2 ]
&\le&
(Const.)
{\bf E}[  | Y_{L, L}(\kappa) |^2 ]
=
O(L^{- 2\alpha+1}).
\eeq
A standard argument 
using Chebishev's inequality yields the conclusion.
The proof for 
$\kappa = 0$
is similar except that 
$\langle F \rangle = 0$ 
and 
$II_{t, L} = 0$.
\\
(2)
It easily follows from the argument above. 
\QED
\end{proof}
\begin{lemma}
\label{thetatilde}
For 
$\kappa \ge 0$, 
\beq
\sup_{0 \le t \le L} 
|\tilde{\theta}_{t, L}(\kappa)| 
\stackrel{L \to \infty}{\to} 
0
\eeq
in probability.
\end{lemma}
\begin{proof}
By
(\ref{theta-eq})
$
\tilde{\theta}_{t, L}(\kappa)
=
\frac {1}{2 \kappa}
Re
\left(
J_{t, L}(\kappa) - J_{t, L}(0)
\right)
$. 
Then 
the conclusion follows from 
Lemma \ref{J1}(1).
\QED
\end{proof}
{\it Proof of Theorem \ref{clock}}\\
By
Lemma \ref{thetatilde} 
and (A), 
\beq
\Psi_{L,L}(x)
&=&
x + \tilde{\theta}_{L,L}(\kappa_0 + \frac xL)
-
\tilde{\theta}_{L,L}(\kappa_0)
=
x + o(1)
\\
\theta_{L_j, L_j }(\kappa_0)
&=&
\kappa_0 L_j 
+
\tilde{\theta}_{L_j, L_j}(\kappa_0)
=
m_j \pi + \beta + o(1)
\eeq
in probability.
Hence for any subsequence of 
$\{ L_j \}$, 
we can further find a subsequence 
$\{ L_{j_k} \}$
of that such that 
\[
\lim_{k \to \infty}
\phi(\kappa_0, L_{j_k})
=
\beta, 
\quad
a.s.
\]
By using the fact that 
$
\lim_{k \to \infty} (n(L_{j_k}) - m_{L_{j_k}})
=
- \infty
$
and 
Lemma \ref{Inverse}, 
we have 
\beq
\lim_{k \to \infty}
{\bf E}[ e^{- \xi_{L_{j_k}}(f)} ]
=
\exp \left(
- \sum_{n \in {\bf Z}}
f 
(n \pi - \beta)
\right). 
\eeq
Since this holds for any subsequence of 
$\{ L_j \}$, 
we arrive at the conclusion. 
\QED
%

\subsection{Behavior of solutions}
We study the behavior of the solution 
$x_t$
to the 
Schr\"odinger equation
$H_L x_t = \kappa^2 x_t$.
\begin{theorem}
\label{free}
For a solution 
$x_{t, L}$
to the equation 
$H_L x_{t, L} = \kappa^2 x_{t, L}$
$(\kappa > 0, x_{0, L} = 0)$, 
let 
$r_{t, L}$, $\theta_{t, L}$
be the corresponding Pr\"ufer variables defined in 
(\ref{Prufer}).
Then we have 
\beq
&&\sup_{0 \le t \le L} | r_{t, L}(\kappa) - 1 |
\to_P 0
\\
&&\sup_{0 \le t \le L} 
| \theta_{t, L}(\kappa) - \kappa t | 
\to_P 0
\eeq
in probability so that 
$x_{t, L}$
approaches to the free solution. 
\end{theorem}
\begin{proof}
It easily follows from 
(\ref{r-eq}), (\ref{theta-eq})
and Lemma \ref{J1}(1).
\QED
\end{proof}
%
\subsection{Second Limit Theorem}
\begin{lemma}
\label{Eigenvalue}
\[
\sqrt{E_{m_j+n} (L_j)} 
=
\kappa_0 + 
\frac {n\pi - \beta + o(1)}{L_j}
\]
in probability.
\end{lemma}
\begin{proof}
This lemma follows from Lemma \ref{thetatilde} and the following computation. 
\beq
&&
\left(
\sqrt{E_{m_j+n}(L_j)} - \kappa_0
\right)
L_j 
\\
&=&
\theta_{L_j, L_j}
\left(
\sqrt{E_{m_j+n}(L_j)}
\right)
-
\kappa_0 L_j
-
\tilde{\theta}_{L_j, L_j}
\left(
\sqrt{E_{m_j+n}(L_j)}
\right)
\\
&=&
(m_j+n) \pi 
- 
\left(
m_j \pi + \beta + o(1)
\right).
\eeq
\QED
\end{proof}
For 
$c_1, c_2, d_1, d_2 \in {\bf R}$, 
set 
\beq
\kappa_c &:=& \kappa_0 + \frac cn, 
\quad
c \in {\bf R}
\\
\Theta_t^{(n)}(c_1, c_2)
 &:=& 
\tilde{\theta}_{nt, n}
\left(
\kappa_{c_1}
\right)
-
\tilde{\theta}_{nt, n}
\left(
\kappa_{c_2}
\right)
\\
C_t(c_1, c_2 ; d_1, d_2)
&:=&
\frac {
C(E_0) }
{8 E_0}
Re 
\int_0^t
(e^{2i c_1 u} - e^{2i c_2 u})
(e^{-2i d_1 u} - e^{-2i d_2 u})
du.
\eeq
We study the behavior of 
$\Theta_t^{(n)}(c_1, c_2)$
as 
$n$
tends to infinity, for fixed
$c_1, c_2$. 
\begin{lemma}
\label{Gauss}
As 
$n \to \infty$, 
$\{ n^{\alpha - \frac 12} \Theta_t^{(n)}(c_1, c_2) \}_{t, c_1, c_2}$
converges to the Gaussian system
$\{ G(t, c_1, c_2) \}_{t, c_1, c_2}$
with covariance 
$\{ C_{t \wedge t'}
(c_1, c_2 ; d_1, d_2) \}_{t, c_1, c_2, d_1, d_2}$.
\end{lemma}
\begin{proof}
Set 
$\kappa_j := \kappa_0 + \frac {c_j}{n}$, 
$\kappa'_j := \kappa_0 + \frac {d_j}{n}$, 
$j=1,2$. 
By
Lemma \ref{J1}(2)
\beq
n^{\alpha - \frac 12}
\Theta_t^{(n)}(c_1, c_2)
&=&
n^{\alpha - \frac 12}
\frac {1}{2 \kappa}
Re
\left(
J_{nt, n}(\kappa_1) - J_{nt, n}(\kappa_2)
\right)
+
O(n^{- \frac 12})
\\
&=&
n^{\alpha - \frac 12}
\frac {1}{2 \kappa}
Re
\left(
Y_{nt, n}(\kappa_1)- Y_{nt, n}(\kappa_2)
\right)
+ o(1)
\eeq
where 
$Y_{t,n}(\kappa)$
is defined in 
(\ref{Y}).
Set 
\[
Z^{(n)}_t(\kappa_1, \kappa_2) := n^{\alpha - \frac 12} 
\left(
Y_{nt, n}(\kappa_1) - Y_{nt, n}(\kappa_2)
\right).
\]
By Lemma \ref{PartialIntegration} we have
\beq
\langle Z^{(n)}(\kappa_1, \kappa_2), 
Z^{(n)}(\kappa'_1, \kappa'_2) 
\rangle_{t}
&=&
o(1)
\\
\langle Z^{(n)}(\kappa_1, \kappa_2), 
\overline{Z^{(n)}(\kappa'_1, \kappa'_2)}
\rangle_{t}
&=&
C(E_0)
\int_0^t 
\left(
e^{2i c_1 u} - e^{2i c_2 u}
\right)
\left(
e^{-2i d_1 u} - e^{-2i d_2 u}
\right)
du
+
o(1).
\eeq
It then suffices 
to use the martingale central limit theorem.
\QED
\end{proof}
{\it Proof of Theorem \ref{2ndLimitTheorem}}\\
Let 
$b_n$
such that 
$
\sqrt{E_{m_j+n}(L_j)}
=
\kappa_0
+
\frac {b_n}{L_j}.
$
Then by Lemma \ref{Eigenvalue}, 
$
b_n = n \pi - \beta + o(1)
$
in probability.
Taking difference between 
\beq
(m_j + n + 1)\pi
&=&
\sqrt{E_{m_j+n+1}(L_j)} L_j 
+
\tilde{\theta}_{L_j, L_j}
\left(
\sqrt{E_{m_j+n+1}(L_j)}
\right)
\\
(m_j + n )\pi
&=&
\sqrt{E_{m_j+n}(L_j)} L_j 
+
\tilde{\theta}_{L_j, L_j}
\left(
\sqrt{E_{m_j+n}(L_j)}
\right)
\eeq
yields
\beq
\left\{
\left(
\sqrt{E_{m_j+n+1}(L_j)}
-
\sqrt{E_{m_j+n}(L_j)}
\right)
L_j
-
\pi
\right\}
L_j^{\alpha - \frac 12}
=
- L_j^{\alpha - \frac 12}
\Theta_1^{(L_j)}(b_{n+1}, b_n).
\eeq
By using 
Lemma \ref{Gauss}
and 
Skorohard's theorem, we obtain the conclusion. 
The statement of covariance
follows from the following computation.  
\beq
&&
Re 
\int_0^1 
\left(
e^{2i ((n+1) \pi - \beta)s}
-
e^{2i (n \pi - \beta)s}
\right)
\left(
e^{-2i ((n'+1) \pi - \beta)s}
-
e^{-2i (n' \pi - \beta)s}
\right)
ds
\\
&=&
\frac {1}{\pi}
\int_0^{2 \pi}
\cos ((n-n') \theta) 
(1 - \cos \theta) 
d \theta.
\eeq
\QED

\section{Critical Case}
In this section we set
$\alpha = \frac 12$.
%
\subsection{Preliminaries}
\begin{lemma}
\label{Apriori1}
Let 
$J_{t, L}(\kappa)$, $\kappa \ge 0$
be the one defined in (\ref{Jdefine}). \\
(1)
For 
$\kappa > 0$ : 
\beq
J_{t, L}(\kappa)
&=&
- \frac {i}{2 \kappa}
\langle F g_{\kappa} \rangle \cdot 
\lambda_L^2 \cdot t
+
Y_{t, L}(\kappa)
+
Z_{t, L}(\kappa)
+
O(\lambda_L) 
+
O(\lambda_L^3 \cdot t)
\eeq
where 
$Y_{t, L}(\kappa)$, 
$Z_{t, L}(\kappa)$
are martingales such that 
\beq
Y_{t, L}(\kappa)
&:=&
- \frac {2i}{2 \kappa} 
\cdot
\lambda_L^2
\left(
\frac 12 K_{4, 3}(\kappa) + \frac 12 K_{0, 3}(\kappa)
- K_{2, 3}(\kappa)
\right)
\\
K_{\beta, 3}(\kappa)
&:=&
\lambda_L^2
\int_0^t 
e^{i \beta \theta_{s, L}(\kappa)}
d \widetilde{M}_s^{(\beta)}(\kappa), 
\quad
\beta = 0, 2, 4
\\
Z_{t, L}(\kappa)
&:=&
\lambda_L
\int_0^t 
e^{2i \theta_{s, L}(\kappa)}
d M_s(\kappa).
\eeq
(2)
For 
$\kappa = 0$ : 
\beq
J_{t, L}(0)
&:=&
Z_{t, L} + O(\lambda_L), 
\quad
Z_{t, L}(0) := \lambda_L M_t
\eeq
where 
$g_{\kappa}$,  
$M_s(\kappa)$, 
$M_s$, 
$\widetilde{M}_s^{(\beta)}(\kappa)$
are defined in 
Lemma \ref{PartialIntegration}.
\end{lemma}
\begin{proof}
(1)
By
Lemma \ref{PartialIntegration}
\beq
J_{t, L}(\kappa)
&=&
\lambda_L
\Biggl\{
\left[
e^{2i \theta_{s, L}(\kappa)} g_{\kappa}(X_s)
\right]_0^t
\\
&& - 
\frac {2i}{2 \kappa}
\int_0^t 
Re
\left(
e^{2i \theta_{s, L}(\kappa)}-1
\right)
e^{2i \theta_{s, L}(\kappa)}
\lambda_L F(X_s) g_{\kappa}(X_s) ds
\\
&& +
\int_0^t 
e^{2i \theta_{s, L}(\kappa)}
dM_s (\kappa) 
\Biggr\}
\\
&=:&
I + II + III.
\eeq
Clearly, 
$I = O(\lambda_L)$ 
and 
$III = Z_{t, L}( \kappa )$. 
For the second term $II$, 
\beq
II
&=&
- \frac {2i}{2 \kappa}
\lambda_L^2
\int_0^t
\left(
\frac {
e^{4i \theta_{s, L}(\kappa)}+1
}
{2}
-e^{2i \theta_{s, L}(\kappa)}
\right)
F(X_s) g_{\kappa}(X_s) ds
\\
&=:&
- \frac {2i}{2 \kappa}
\left(
\frac 12 K_4 (\kappa) + \frac 12 K_0 (\kappa)
- K_2 (\kappa)
\right)
\eeq
where 
\beq
K_{\beta}(\kappa)
&:=&
\lambda_L^2
\int_0^t 
e^{i \beta \theta_{s, L}(\kappa)}
F(X_s) g_{\kappa}(X_s) ds, 
\quad
\beta = 0, 2, 4.
\eeq
For 
$\beta = 2, 4$
we use Lemma \ref{PartialIntegration}
again 
\beq
K_{\beta}(\kappa)
&=&
\left[
\lambda_L^2 
e^{i \beta \theta_{s, L}(\kappa)}
h_{\kappa, \beta}(X_s)
\right]_0^t
\\
&& - 
\frac {i \beta}{2 \kappa}
\cdot
\lambda_L^2
\cdot
\int_0^t
Re \left(
e^{2i \theta_{s, L}(\kappa)}-1
\right)
e^{i \beta \theta_{s, L}(\kappa)}
\lambda_L
F(X_s) g_{\kappa}(X_s) ds
\\
&& +
\lambda_L^2
\int_0^t 
e^{i \beta \theta_{s, L}(\kappa)}
d \widetilde{M}_s^{(\beta)}(\kappa)
\\
&=:&
K_{\beta, 1} + K_{\beta, 2} + K_{\beta, 3}.
\eeq
We have 
$K_{\beta, 1} 
= 
O(\lambda_L^2) = O(L^{- 2 \alpha})$, 
$K_{\beta, 2}
=
O(\lambda_L^3 \cdot t)
=
O(L^{-3 \alpha} \cdot t)$.
Similarly for 
$\beta = 0$, 
\beq
K_0 (\kappa)
&=&
\lambda_L^2\langle F g_{\kappa} \rangle \cdot t
+
\lambda_L^2
\left[ h_{0,\kappa}(X_s) \right]_0^t
+
\lambda_L^2
\widetilde{M}_t
\\
&=:&
K_{0, 1} + K_{0,2} + K_{0, 3}.
\eeq
We then have 
$K_{0, 2}
=
O(\lambda_L^2)
=
O(L^{-2 \alpha})$. 
Putting together 
\beq
II
&=&
- \frac {2i}{2 \kappa} \cdot
\lambda_L^2 \langle F g_{\kappa} \rangle \cdot t
- \frac {2i}{2 \kappa} 
\cdot
\lambda_L^2
\left(
\frac 12 K_{4, 3} + \frac 12 K_{0, 3}
- K_{2, 3}
\right)
\\
&& \qquad
+
O(\lambda_L^2) 
+
O(\lambda_L^3 \cdot t)
\eeq
proving (1). \\
(2)
It immediately follows from 
Lemma \ref{PartialIntegration} and the fact that 
$\langle F \rangle = 0$.
\QED
\end{proof}
\begin{lemma}
\label{Exponential}
For any 
$\gamma > 0$, 
\beq
{\bf E}\left[
\sup_{0 \le t \le L}
e^{\gamma (Y_{t, L}(\kappa) + Z_{t, L}(\kappa))}
\right]
\le 
C_{\gamma}.
\eeq
\end{lemma}
\begin{proof}
By Ito's formula 
\begin{eqnarray}
&&
e^{\gamma (Y_{t,L}(\kappa)+Z_{t,L}(\kappa))}
\nonumber
\\
&=&
1 
+
\gamma
\int_0^t 
e^{\gamma(Y_{s,L}(\kappa)+Z_{s,L}(\kappa))}
d (Y_{\cdot, L}+ Z_{\cdot, L})_s
\nonumber
\\
&& +
\frac {\gamma^2}{2}
\int_0^t 
e^{\gamma(Y_{s,L}(\kappa)+Z_{s,L}(\kappa))}
d \langle Y_{\cdot, L}+ Z_{\cdot, L},
Y_{\cdot, L}+ Z_{\cdot, L} \rangle_s
\nonumber
\\
&& +
\frac {\gamma^2}{2} 
\int_0^t 
e^{\gamma(Y_{s,L}(\kappa)+Z_{s,L}(\kappa))}
d \langle Y_{\cdot, L}+ Z_{\cdot, L},
\overline{
Y_{\cdot, L}+ Z_{\cdot, L}
} 
\rangle_s
\label{Exponential-ex}
\end{eqnarray}
we note that
\beq
&&
| 
d \langle Y_{\cdot, L}+ Z_{\cdot, L},
Y_{\cdot, L}+ Z_{\cdot, L} \rangle_s
|
\le
(Const.) 
(\lambda_L^2 + \lambda_L^3 + \lambda_L^4) ds
\\
&&
| 
d \langle Y_{\cdot, L}+ Z_{\cdot, L},
\overline{
Y_{\cdot, L}+ Z_{\cdot, L}
} \rangle_s
|
\le
(Const.) 
(\lambda_L^2 + \lambda_L^3 + \lambda_L^4) ds.
\eeq
Setting 
\beq
f_{\gamma}(t) 
&:=&
{\bf E}
\left[
e^{ \gamma(Y_{t,L}(\kappa)+Z_{t,L}(\kappa)) }
\right]
\eeq
we have by 
(\ref{Exponential-ex}), 
\beq
f_{\gamma}(t)
& \le &
1 + C_L \int_0^t f_{\gamma}(s) ds, 
\quad
C_L
:=
C\left(
\lambda_L^2 + \lambda_L^3 + \lambda_L^4
\right).
\eeq
By Grownwall's inequality
$
f_{\gamma}(t) \le e^{C_L \cdot t}
$
yielding 
\beq
\sup_{0 \le t \le L} f_{\gamma}(t)
\le
e^{C}.
\eeq
By martingale inequality, 
\beq
&&
{\bf E}\left[
\sup_{0 \le t \le L}
\left|
\int_0^t 
e^{\gamma(Y_{s,L}(\kappa)+Z_{s,L}(\kappa))}
d \left(
Y_{\cdot, L}+ Z_{\cdot, L}
\right)_s
\right|^2
\right]
\\
& \le &
(Const.)
{\bf E}\left[
\left|
\int_0^L 
e^{\gamma(Y_{s,L}(\kappa)+Z_{s,L}(\kappa))}
d \left(
Y_{\cdot, L}+ Z_{\cdot, L}
\right)_s
\right|^2
\right]
\\
&\le&
(Const.)
{\bf E}\left[
\int_0^L
e^{2\gamma(Y_{t,L}(\kappa)+Z_{t,L}(\kappa))}
d \langle 
Y_{\cdot, L}+ Z_{\cdot, L},
Y_{\cdot, L}+ Z_{\cdot, L}
\rangle_s
\right]
\\
&&
+(Const.)
{\bf E}\left[
\int_0^L
e^{2\gamma(Y_{t,L}(\kappa)+Z_{t,L}(\kappa))}
d \langle 
Y_{\cdot, L}+ Z_{\cdot, L},
\overline{
Y_{\cdot, L}+ Z_{\cdot, L}
}
\rangle_s
\right]
\\
& \le &
(Const.)
\int_0^L
f_{2\gamma}(s) \lambda_L^2 ds
\le
(Const.).
\eeq
Substituting to 
(\ref{Exponential-ex}), we arrive at the conclusion.
\QED
\end{proof}
%
\subsection{Proof of Theorem \ref{Schtau}}
We set 
\beq
\Psi_L(c)
&:=&
2 c + 2 \tilde{\theta}_{L, L}
\left(
\kappa_0 + \frac cL
\right)
=
2 \theta_{L, L}
\left(
\kappa_0 + \frac cL
\right)
-
2 \kappa_0 L.
\eeq
By definition, 
$\Psi_L(c)$
is increasing with respect to 
$c$
for large 
$L$. 
As Lemma \ref{Laplace}
we have 
\begin{lemma}
Writing 
$\kappa_0 L_j$
as 
\[
\kappa_0 L_j
=
m_j \pi +  \beta_j, 
\quad
m_j \in {\bf Z}, 
\quad
\beta_j \in [0, \pi), 
\]
we have
\beq
{\bf E}[ e^{- \xi_{L_j} (f)} ]
=
{\bf E}
\left[
\exp \left(
- \sum_{n \ge n(L_j) - m_j}
f 
\left(
\Psi_{L_j}^{-1}
(2 n \pi - 2 \beta_j)
\right)
\right)
\right].
\eeq
\end{lemma}
It then suffices to study the behavior of 
$\Psi_{L_j}(c)$
as 
$j \to \infty$. 
Replacing 
$L_j$
by 
$n$, 
we set 
\beq
\kappa_c 
&:=& 
\kappa_0 + \frac cn, 
\quad
c \in {\bf R}
\\
\Psi_t^{(n)}(c)
&:=&
2 ct + 2 \tilde{\theta}_{nt, n}
\left(
\kappa_c 
\right), 
\quad
0 \le t \le 1. 
\eeq
For simplicity we set 
$
J^{(n)}_t(\kappa)
:=
J_{nt, n}(\kappa).
$
By 
(\ref{theta-eq}), 
we have 
\beq
\tilde{\theta}_{nt, n}(\kappa_c)
&=&
\frac {1}{2 \kappa}
Re \;
\left(
J_t^{(n)}(\kappa_c) - J_t^{(n)}(0)
\right)
+ O( \lambda_n ).
\eeq
By 
Lemma \ref{Apriori1}
and the fact that 
$\langle Y_{\cdot, n},Y_{\cdot, n} \rangle_t$, 
$\langle Y_{\cdot, n}, \overline{Y_{\cdot, n}} \rangle_t
=
O(\lambda_n^4 \cdot nt)
=
O(n^{-1})$, 
we have 
\beq
J_t^{(n)}(\kappa_c)
-
J_t^{(n)}(0)
&=&
W_t^{(n)}(c)
- \frac {i}{2 \kappa_0}
\langle F g_{\kappa_0} \rangle t
+
o(1)
\eeq
in probability, where
\beq
W_t^{(n)}(c)
&:=&
Z_{nt, n}(\kappa_c) - Z_{nt, n}(0).
\eeq
Set 
$\varphi := [ g, g ]$, 
$\varphi_{\kappa_0} := [ g_{\kappa_0}, g_{-\kappa_0} ]$.
Then by Lemma \ref{PartialIntegration}
\begin{eqnarray}
\langle
W^{(n)}(c), W^{(n)}(d)
\rangle_t
&=&
\langle \varphi \rangle t + o(1)
\label{WW}
\\
\langle W^{(n)}(c), \overline{ W^{(n)}(d) }\rangle_t
&=&
\langle \varphi_{\kappa_0} \rangle
\int_0^t 
e^{i (\Psi_u^{(n)}(c) - \Psi_u^{(n)}(d))}
du
+ \langle \varphi \rangle t + o(1)
\label{WWbar}
\end{eqnarray}
in probability.
The following estimate
\beq
{\bf E} \left[
| W_t^{(n)}(c) - W_s^{(n)}(c) |^2 
\right]
&=&
{\bf E}\left[
| \langle 
W_t^{(n)}(c), \overline{W_t^{(n)}(c)} \rangle_t
-
\langle W_s^{(n)}(c), \overline{W_s^{(n)}(c)} \rangle_s |^2
\right]
\\
& \le &
\lambda_n^2 \langle \varphi_{\kappa_0} \rangle
\int_{ns}^{nt} du
+
\langle \varphi \rangle (t-s) + o(1)
\\
& \le &
C(t-s) + o(1), 
\eeq
together with martingale inequality and Kolmogorov's theorem implies that the sequence of processes
$(W_t^{(n)}(c))_{0 \le t \le 1}$
is tight. 
Therefore by taking subsequences further 
we may assume  
\beq
W_t^{(n)}(c) \to W_t (c)
\quad
\Psi_t^{(n)}(c) \to \Psi_t(c), 
\quad
a.s.
\eeq
Letting
$n \to \infty$
in 
(\ref{WW}), (\ref{WWbar}), 
martingale 
$W(c)$
satisfies
\beq
\langle W(c), W(d) \rangle_t
&=&
\langle \varphi \rangle t
\\
\langle W(c), \overline{W(d)} \rangle_t
&=&
\langle \varphi_{\kappa_0} \rangle
\int_0^t 
e^{i (\Psi_u(c) - \Psi_u (d))} du
+
\langle \varphi \rangle t
\eeq
so that we have 
\beq
W_t(c)
=
\sqrt{
\frac {\langle \varphi_{\kappa_0} \rangle}{2}
}
\int_0^t e^{i \Psi_s(c)} d Z_s
+
\sqrt{ \langle \varphi \rangle }
B_t
\eeq
where 
$Z_t$, $B_t$
are mutually independent, complex and standard Brownian motions respectively. 
Since 
\[
\Psi_t(c) = 2ct 
+
\frac {1}{\kappa_0} Re \; 
\left(
W_t(c) - \frac {i}{2 \kappa_0} \cdot t \cdot 
\langle F g_{\kappa_0} \rangle
\right)
\]
we have 
\beq
d \Psi_t (c)
&=&
2c dt + 
\frac {1}{\kappa_0}
\Biggl\{
\sqrt{
\frac {\langle \varphi_{\kappa_0} \rangle}{2}
}
Re \; 
\left(
e^{i \Psi_s(c)} d Z_s
\right)
+
\sqrt{ \langle \varphi \rangle }
d B_t
- Re \; 
\frac {i}{2\kappa_0}\langle F g_{\kappa_0} \rangle
dt
\Biggr\}.
\eeq
Similar arguments show that 
$\Psi_t(c_1), \cdots, \Psi_t(c_m)$
jointly satisfy the SDE (\ref{SchtauSDE}). 
This determines the process
$\Psi_t(c)$
uniquely, which is strictly-increasing by SDE comparison theorem.
That 
$\Psi_t^{(n)}(c) \to \Psi_t(c)$
also in the sense of strictly-increasing function valued process 
follows from \cite{KN}, Proposition 9.2. 
By Lemma \ref{Inverse}
we finish the proof of Theorem \ref{Schtau}.
%

%
\subsection{Behavior of Solutions for Critical Case}
\begin{theorem}
\label{Delocalized}
Suppose that 
$x_t$
is the solution to the Schr\"odinger equation : 
$H x_t = \kappa^2 x_t$, $\kappa > 0$, 
%
and let 
$r_t(\kappa), \theta_t(\kappa)$
be the corresponding Pr\"ufer variables. 
Then we can find 
$C_{\kappa} < \infty$
such that for any 
$a>0$, 
\beq
{\bf P}\left(
\frac 1a \le r_t \le a
\;
\mbox{ for any }
t \in [0,L]
\right)
\ge 
1 - \frac {C_{\kappa}}{a}.
\eeq
\end{theorem}
%
%
\begin{proof}
By
(\ref{r-eq}), 
Lemma \ref{Apriori1}, \ref{Exponential}, 
\beq
\log r_t &=& B(t) + Y(t) + Z(t), 
\eeq
where
\beq
&&
M:=\sup_{0 \le t \le L}| B(t) | < \infty, 
\quad
{\bf E}[ \sup_{0 \le t \le L} e^{Y_t + Z_t} ] 
 \le  C_1, 
\quad
{\bf E}[ \sup_{0 \le t \le L} e^{-(Y_t + Z_t)} ] 
\le  C_2. 
\eeq
Then 
the conclusion follows from Chebyshev's inequality. 
\QED
\end{proof}
%

\section{Decaying potentials with critical rate}
For the proof of 
Theorem \ref{Carousel}, 
we solve the eigenvalue equation
$H x_t = k^2 x_t$
with 
$x_L=0$, $x'_L / \kappa = 1$. 
Here we suppose that the distribution of 
$X_0$ 
is uniform on 
$M$.
Then 
$(X_t)_{t \in {\bf R}}$
is stationary, so that if 
$( Y_t )_{t \in {\bf R}}$
is a independent copy of 
$(X_t)$,  
we may replace 
$H_n$ 
by the following operator. 
\[
\hat{H}_n
:= - \frac {d^2}{d t^2} + \hat{V}(t), 
\quad
\hat{V}(t) := a(n-t) F(Y_t)
\quad
\mbox{ on }
\quad
L^2 (0,n).
\]
Moreover by \cite{KN}, 
$\xi_{\infty} = \lim_{n \to \infty}\xi_n$
is uniquely determined as far as 
$a(t) = t^{- \frac 12}(1+o(1))$. 
Therefore, without loss of generality, we may suppose that 
\beq
a(s) = 
\frac {1}{\sqrt{s}}, 
\quad 
s \ge 1.
\eeq
Furthermore we set 
\beq
\kappa_0 &:=& \sqrt{E_0}, \quad
\kappa_c := \kappa_0 + \frac cn, 
\quad
c \in {\bf R}
\\
\Psi_t^{(n)}(c)
&:=&
2 \theta_{nt} 
\left(
\kappa_c
\right) - 
2 \theta_{nt} (\kappa_0).
\eeq
%
\subsection{A priori estimate}
In this subsection we show the following theorem. 
\begin{theorem}
\label{Apriori}
For any fixed 
$T < 1$
we have
\beq
\Psi_t^{(n)}(c)
&=&
2 c t 
+ 
\frac {1}{2 \kappa_0} Re \; V_t^{(n)}(c)
+
\delta_{nt}(\kappa_c) - \delta_{nt}(\kappa_0)
+
O(n^{- \frac 12}), 
\quad
0 \le t \le T
\eeq
where 
\beq
V_t^{(n)}(c)
&:=&
Y_t^{(n)}(\kappa_c) - V_t^{(n)}(\kappa_0)
\\
Y_t^{(n)}(\kappa)
&:=&
\int_0^{nt} a(n-s) e^{2i \theta_s(\kappa)} dM_s(\kappa).
\eeq
The concrete form of 
$\delta_{nt}(\kappa)$
is given in Lemma \ref{J} below.
Moreover
\beq
&&
{\bf E} \left[
\sup_{0 \le t \le T}
| V_t^{(n)}(c) |
\right]
\le
C 
\frac {1}{\sqrt{1-T}}
\\
&&
{\bf E} \left[
\sup_{0 \le t \le T}
| \delta_{nt}(\kappa_c) - \delta_{nt}(\kappa_0) |^2
\right]
\stackrel{n \to \infty}{\to} 0.
\eeq
\end{theorem}
First of all, by 
(\ref{theta-eq})
it is easy to see 
\begin{eqnarray}
\Psi_t^{(n)}(c)
&=&
2ct + 
\frac {1}{\kappa_0} Re \;
\left(
J_t^{(n)}(\kappa_c) - J_t^{(n)}(\kappa_0)
\right)
+
O(n^{-\frac 12})
\label{Decomposition}
\end{eqnarray}
where
\beq
J_t^{(n)}(\kappa)
&:=&
\int_0^{nt}
e^{2i \theta_s(\kappa)} a(n-s) F(Y_s) ds.
\nonumber
\eeq
We decompose this integral by using Lemma \ref{PartialIntegration}. 
The result is :  
%
\begin{lemma}
\label{J}
\beq
J_t^{(n)}(\kappa)
&=&
C_t^{(n)}(\kappa)
+
\delta_{nt}(\kappa)
+ Y_t^{(n)}(\kappa), 
\quad
\kappa > 0
\eeq
where 
\beq
C_t^{(n)}(\kappa)
&:=&
- \frac {i}{2 \kappa}
\int_0^{nt}
a(n-s)^2 F(Y_s) g_{\kappa}(Y_s) ds
\\
\delta_{nt}(\kappa)
&:=&
\delta_{nt}^{(1)}(\kappa) + \delta_{nt}^{(2)} (\kappa)
\\
\delta_{nt}^{(1)}(\kappa)
&:=&
\left[
a(n-s) e^{2i \theta_s(\kappa)} g_{\kappa}(Y_s)
\right]_0^{nt}
- \int_0^{nt}
( a(n-s) )' e^{2i \theta_s(\kappa)} g_{\kappa}(Y_s) ds
\\
\delta_{nt}^{(2)}(\kappa)
&:=&
\frac {i}{\kappa}
\int_0^{nt}
\left(
\frac {e^{2i \theta_s(\kappa)}}{2} -1 
\right)
e^{2i \theta_s(\kappa)}
a(n-s)^2 F(Y_s) g_{\kappa}(Y_s) ds
\\
Y_t^{(n)}(\kappa)
&:=&
\int_0^{nt} a(n-s) e^{2i \theta_s(\kappa)} dM_s(\kappa).
\eeq
\end{lemma}
To compute 
$J_t^{(n)}(\kappa_c) - J_t^{(n)}(\kappa_0)$
we estimate the difference of them : 
\begin{lemma}
\label{CDelta}
For 
$0 \le t < 1$
\beq
&(1)& \quad
\left|
C_t^{(n)}(\kappa_c) - C_t^{(n)}(\kappa_0)
\right|
\le
C \,\frac {\log n}{n}
\\
&(2)& \quad
\delta^{(1)}_{nt}(\kappa_c) - \delta^{(1)}_{nt}(\kappa_0)
=
O(n^{- \frac 12}(1-t)^{-\frac 12})
\\
&(3)&\quad
\delta_{nt}^{(2)}(\kappa)
=
- \frac {i}{2 \kappa} Z_4^{(n)}(\kappa) + \frac {i}{\kappa} Z_2^{(n)}(\kappa)
+
O(n^{-\frac 12}(1-t)^{-\frac 12})
\eeq
where
\beq
Z_{\beta}^{(n)}(\kappa)
&:=&
\int_0^{nt} a(n-s)^2 e^{i \beta \theta_s(\kappa)} 
d \tilde{M}_s (\kappa).
\eeq
\end{lemma}
\begin{proof}
It is sufficient to show (3). 
\beq
\delta_{nt}^{(2)}(\kappa)
&=&
- \frac {i}{2 \kappa} D_4^{(n)}(\kappa) + \frac {i}{\kappa} D_2^{(n)}(\kappa)
\eeq
where we set 
\beq
D_{\beta}^{(n)}(\kappa)
&:=&
\int_0^{nt}
a(n-s)^2 e^{i \beta \theta_s(\kappa)} 
F(Y_s) g_{\kappa} (Y_s) ds, 
\quad
\beta = 2, 4.
\eeq
Thus it suffices to estimate
$D_{\beta}(\kappa)$.
By
Lemma \ref{PartialIntegration}
\beq
D_{\beta}^{(n)}(\kappa)
&=&
\left[
a(n-s)^2 e^{i \beta \theta_s(\kappa)} h_{\kappa, \beta}(Y_s) 
\right]_0^{nt}
\\
&& - 
\int_0^{nt} ( a(n-s)^2 )' e^{i \beta \theta_s(\kappa)}
h_{\kappa, \beta}(Y_s) ds
\\
&& - 
\frac {i \beta}{2 \kappa}
\int_0^{nt}
Re \left( e^{2i \theta_s(\kappa)} - 1 \right)
e^{i \beta \theta_s(\kappa)} 
a(n-s)^3 F(Y_s) h_{\kappa, \beta}(Y_s) ds
\\
&&+
\int_0^{nt} a(n-s)^2 e^{i \beta \theta_s(\kappa)} 
d \widetilde{M}_s^{(\beta)} (\kappa).
\eeq
Since
$
\int_0^{nt} 
(a(n-s)^2)' ds
=
O(n^{-1}(1-t)^{-1})
$
and
$
\int_0^{nt} 
a(n-s)^3 ds
=
O(n^{-\frac 12}(1-t)^{-\frac 12})
$, 
we have 
\beq
D_{\beta}^{(n)}(\kappa)
=
Z_{\beta}^{(n)}(\kappa)
+
O(n^{- \frac 12}(1-t)^{-\frac 12}).
\eeq
\QED
\end{proof}
We next estimate
$V_t^{(n)}(c)$.
\begin{lemma}
\label{V}
(1)
If 
$0 < t < 1$, 
we have
\beq
{\bf E}\left[ 
\langle 
V^{(n)}(c), \overline{V^{(n)}}(c) 
\rangle_t
\right]
\le
\frac {C}{1-t}
\eeq
(2)
For fixed 
$0 < T < 1$, 
we have
\[
{\bf E}\left[ 
\sup_{0 \le t \le T} | V_t^{(n)}(c) | 
\right]
\le
(const.)
\frac {C}{\sqrt{1 - T}}
\]
\end{lemma}
\begin{proof}
\beq
&&\langle V^{(n)}(c), \overline{V^{(n)}(c)} \rangle_t
\\
&=&
\int_0^{nt}
a(n-s)^2 
\left|
e^{2i (\theta_s(\kappa_c) - \theta_s(\kappa_0))} - 1
\right|^2 
[ g_{\kappa_0}, \overline{g}_{\kappa_0} ]ds
+
O( n^{- \frac 12} (1-t)^{- \frac 12} )
\eeq
Set
$
\varphi_{\kappa_0} := [ g_{\kappa_0}, \overline{g_0}_{\kappa} ]
$.
We compute 
by using Lemma \ref{PartialIntegration}
\beq
&&\langle V^{(n)}(c), \overline{V^{(n)}(c)} \rangle_t
\\
&=&
\langle \varphi_{\kappa_0} \rangle
\int_0^{nt} a(n-s)^2 
\left|
e^{2i (\theta_s(\kappa_c) - \theta_s(\kappa_0))} - 1
\right|^2 
ds
+
W_t^{(n)}(\kappa_0) 
+
O( n^{- \frac 12} (1-t)^{- \frac 12} )
\eeq
where
\beq
&&
W_t^{(n)}(\kappa_0)
:=
\int_0^{nt} 
a(n-s)^2 \left|
e^{2i (\theta_s(\kappa_c) - \theta_s(\kappa_0))} - 1
\right|^2 
d M_s(\varphi_{\kappa_0}, 0)
\\
&&
\langle W^{(n)}(\kappa_0), W^{(n)}(\kappa_0) \rangle_t
=
O(n^{-1}(1-t)^{-1}).
\eeq
Taking expectation, 
martingale term vanishes and we have 
\beq
&&
{\bf E} \left[
\langle V^{(n)}(c), \overline{V^{(n)}}(c) \rangle_t
\right]
\\
&=&
\langle \varphi_{\kappa_0} \rangle
\int_0^{nt} a(n-s)^2
{\bf E} \left[
\left|
e^{2i (\theta_s(\kappa_c) - \theta_s(\kappa_0))}-1
\right|^2
\right]
ds
+
O(n^{- \frac 12}(1-t)^{- \frac 12})
\\
& \le &
C
\int_0^t 
\frac {n}{n(1-u)}du
+ 
O(n^{- \frac 12}(1-t)^{- \frac 12})
\\
& \le &
\frac {C}{1-t}
+ 
O(n^{- \frac 12}(1-t)^{- \frac 12}).
\eeq
(2)
follows (1) and the martingale inequality. 
\QED
\end{proof}
\begin{lemma}
\label{Delta}
If 
$0 < T < 1$, 
\beq
{\bf E} \left[
\sup_{0 \le t \le T}
| \delta_{nt}(\kappa_c) - \delta_{nt}(\kappa_0) |^2
\right]
\stackrel{n \to \infty}{\to} 0.
\eeq
\end{lemma}
\begin{proof}
By
Lemma \ref{CDelta}, 
it suffices to show that the martingale part converges to 
$0$, 
that is,
\beq
&&
{\bf E}\left[
\sup_{0 \le t \le T}
\left|
Z_{\beta}^{(n)}(\kappa_c) - Z_{\beta}^{(n)}(\kappa_0)
\right|^2
\right]
\\
& \le &
C\int_0^{nT}
a(n-s)^4 
{\bf E}\left[
\left|
e^{i \beta \theta_s(\kappa_c)}
-
e^{i \beta \theta_s(\kappa_0)}
\right|^2
\right]
ds
+
O( n^{-2} (1-T)^{-2} )
\\
&&
\to 0.
\eeq
\QED
\end{proof}
By combining 
(\ref{Decomposition}), 
Lemma \ref{J}, \ref{CDelta}, \ref{V}
and \ref{Delta}, 
we obtain 
Theorem \ref{Apriori}.
%
\subsection{Tightness}
\begin{lemma}
\label{Kolmogorov}
For
$0 \le s < t  \le T < 1$
we have
\beq
{\bf E} \left[
\left|
V_t^{(n)}(c) - V_s^{(n)}(c) 
\right|^4
\right]
\le (const.) (t-s)^2.
\eeq
\end{lemma}
\begin{proof}
\beq
&&
{\bf E} \left[
\left|
V_t^{(n)}(c) - V_s^{(n)}(c) 
\right|^4
\right]
\\
& \le&
C{\bf E} \left[ | V_t^{(n)}(c) - V_s^{(n)}(c) |^2 \right]^2
\\
& \le &
C
{\bf E}\left[
\int_{ns}^{nt} a(n-u)^2 
\left[
e^{2i \theta_u(\kappa_c)} g_{\kappa_c} - e^{2i \theta_u(\kappa_0)} g_{\kappa_0},
e^{2i \theta_u(\kappa_c)} g_{\kappa_c} - e^{2i \theta_u(\kappa_0)} g_{\kappa_0}
\right]du
\right]^2
\\
& \le &
C
\left(
\int_{ns}^{nt} a(n-u)^2 du 
\right)^2 
\le 
C (t-s)^2.
\eeq
\QED
\end{proof}
By using these lemmas, 
we can show the tightness. 
\begin{theorem}
\label{Tightness}
For any 
$c \in {\bf R}$
$\{ \Psi_t^{(n)}(c) \}_{0 \le t < 1}$
is tight.
In fact, for any 
$0 < T < 1$, 
we have 
\beq
&(1)&\;
\lim_{A \to \infty} {\bf P}\left( 
\sup_{0 \le t \le T}
|\Psi_t^{(n)}(c)| 
\ge A
\right) = 0, 
\\
&(2)&\;
\lim_{\delta \downarrow 0}
\limsup_{n \to \infty}
{\bf P}\left(
\sup_{ | t - s | < \delta,\; 0 \le s, t < T }
| \Psi_t^{(n)}(c) - \Psi_s^{(n)}(c) | > \rho
\right) = 0, 
\quad
\forall \rho > 0. 
\eeq
\end{theorem}
\begin{proof}
(1)
follows from 
Theorem \ref{Apriori}, 
and 
(2)
follows from Theorem \ref{Apriori} and Lemma \ref{Kolmogorov}.
\QED
\end{proof}
%
\subsection{Derivation of SDE}
By
Theorem \ref{Tightness}
$\Psi_t^{(n)}(c)$
have a limit point
$\Psi_t(c)$.
Then by Skorohard's theorem, we may assume
\[
\Psi_t^{(n)}(c) \to \Psi_t(c), 
\quad
a.s.
\]
for some subsequence. 
\begin{theorem}
\label{SDE}
$\Psi_t(c)$
satisfies the following SDE. 
\beq
d \Psi_t (c)
&=&
2c dt + 
\frac {D}{\sqrt{1-t}}
Re
\left[
\left( 
e^{i \Psi_t(c)}-1
\right)
d Z_t
\right], 
\quad
0 \le t < 1
\eeq
where 
$
D :=
\frac {
\sqrt{ \langle \varphi_{\kappa_0} \rangle }
}
{\sqrt{2} \kappa_0}.
$
\end{theorem}
\begin{proof}
By 
Theorem \ref{Apriori}, 
\beq
\Psi_t^{(n)}(c)
&=&
2 ct + \frac {1}{\kappa_0} Re \; V_t^{(n)}(c) + o(1)
\eeq
in probability.
Let 
$c, d \in {\bf R}$.
By 
Lemma \ref{PartialIntegration}
we have
\beq
&&
\langle V^{(n)}(c), V^{(n)}(d) \rangle_t
\\
&=&
\int_0^{nt} 
a(n-s)^2 
\left(
e^{2i \theta_s(\kappa_c)} - e^{2i \theta_s(\kappa_0)}
\right)
\left(
e^{2i \theta_s(\kappa_d)} - e^{2i \theta_s(\kappa_0)}
\right)
\varphi_{\kappa_0}(Y_s) ds
+
O(n^{-1}(1-t)^{-1})
\\
&=&
o(1)
\eeq
\beq
&&
\langle V^{(n)}(c), \overline{V^{(n)}(d)} \rangle_t
\\
&=&
\langle \varphi_{\kappa_0} \rangle
\int_0^{nt}
a(n-s)^2 
\left(
e^{2i ( \theta_s(\kappa_c) - \theta_s(\kappa_0) )}-1
\right)
\left(
e^{-2i ( \theta_s(\kappa_d) - \theta_s(\kappa_0) )}-1
\right)
ds
+ o(1)
\\
&=&
\langle \varphi_{\kappa_0} \rangle
\int_0^{t}
na(n-ns)^2 
\left(
e^{i\Psi_s^{(n)}(c)}-1
\right)
\left(
e^{-i\Psi_s^{(n)}(d)}-1
\right)
ds
+ o(1)
\\
& = &
\langle \varphi_{\kappa_0} \rangle
\int_0^{t}
\frac {1}{1-s}
\left(
e^{i\Psi_s(c)}-1
\right)
\left(
e^{-i\Psi_s(d)}-1
\right)
ds
+ o(1)
\eeq
in probability.
Therefore
\[
V_t(c) := \lim_{n \to \infty} V_t^{(n)}(c)
\]
is a 
$L^2$-continuous martingale
such that 
\beq
&&
\langle V(c), V(d) \rangle_t = 0
\\
&&
\langle V(c), \overline{V(d)} \rangle_t
=
\langle \varphi_{\kappa_0} \rangle
\int_0^t 
\frac {1}{1-s}
\left(
e^{i\Psi_s(c)}-1
\right)
\left(
e^{-i\Psi_s(d)}-1
\right)
ds
\eeq
Hence
$V_t(c)$
satisfies 
\[
d V_t
=
\sqrt{
\frac {
\langle \varphi_{\kappa_0}\rangle
}{2}
}
\frac {1}{\sqrt{1-t}}
\left(
e^{i \Psi_t(c)} - 1
\right)
d Z_t
\]
Since
$
\Psi_t(c)
=
2 ct + \frac {1}{\kappa_0} Re \; V_t(c)
$, 
we are done. 
\QED
\end{proof}
%
%
\subsection{Behavior of 
$\theta_{n-n^{\beta}}$}
Let 
\[
\{ x \}_{2 \pi {\bf Z}} := x - 
\max \{ 2 \pi k \, | \, 
k \in {\bf Z}, 2 \pi k \le x \}.
\]
\begin{theorem}
\label{uniform}
For
$0 < \beta < 1$ 
and 
$\kappa > 0$, 
$\{ 2\theta_{n-n^{\beta}}(\kappa) \}_{2 \pi {\bf Z}}$
converges to the uniform distribution on 
$[0, 2 \pi)$.
\end{theorem}
\begin{proof}
It suffices to show 
\[
\lim_{n \to \infty}
{\bf E}[
e^{2mi \tilde{\theta}_{n-n^{\beta}}(\kappa)}
]
\to 0, 
\quad
m \ne 0.
\]
In what follows, 
we omit the $\kappa$-dependence. 
Set
\[
t =t_n= 1 - n^{\beta -1}.
\]
We then have
\beq
e^{2mi \tilde{\theta}_{n-n^{\beta}}(\kappa)}
&=&
1 + 
\int_0^{nt} 
2mi \frac {1}{2 \kappa}
Re \left( e^{2i \theta_s} - 1 \right)
e^{2mi \tilde{\theta}_s} 
a(n-s) F(Y_s) ds
\\
&=&
1 + \frac{mi}{\kappa}
\int_0^{nt}
\left(
\frac {
e^{2i \kappa s + (2m+2)i \tilde{\theta}_s}
+
e^{-2i \kappa s + (2m-2)i \tilde{\theta}_s}
}
{2}
- e^{2mi \tilde{\theta}_s}
\right)
a(n-s) F(Y_s) ds
\\
&=:&
1 + I + II + III.
\eeq
By
Lemma \ref{PartialIntegration}, 
\beq
I
&=&
\frac {mi}{2 \kappa}
\Biggl\{
\left[
a(n-s) e^{(2m+2)i \tilde{\theta}_s + 2i \kappa s}
g_{\kappa} (Y_s)
\right]_0^{nt}
\\
&& - 
\int_0^{nt}
( a(n-s) )' e^{(2m+2) i \tilde{\theta}_s + 2i \kappa s}
g_{\kappa} (Y_s) ds
\\
&& - 
\frac{(2m+2)i}{2 \kappa}
\int_0^{nt}
a(n-s)^2 
Re \left( e^{2i \theta_s} - 1 \right)
e^{(2m+2)i \tilde{\theta}_s + 2i \kappa s}
g_{\kappa} (Y_s) F(Y_s) ds
\\
&& + 
\int_0^{nt} a(n-s) e^{(2m+2) i \tilde{\theta}_s + 2i \kappa s}
d M_s(\kappa)
\Biggr\}
\\
&=:&
I_1 + I_2 + I_3 + I_4.
\eeq
Since 
$n(1-t) = n^{\beta}$, 
we have 
$
I_1, I_2 = O(n^{- \frac {\beta}{2}}).
$
We further compute 
$I_3$
by using Lemma \ref{PartialIntegration} : 
\beq
I_3 
&=&
\frac {mi}{2 \kappa}
\cdot
\frac {-(2m+2)i}{2 \kappa} 
\cdot
\frac 12
\int_0^{nt}
a(n-s)^2 e^{2mi \tilde{\theta}_s} F(Y_s) g_{\kappa}(Y_s) ds
+
O(n^{- \frac {\beta}{2}})
\\
&=&
\frac {mi}{2 \kappa}
\cdot
\frac {-(2m+2)i}{2 \kappa} 
\cdot
\frac 12
\langle F g_{\kappa} \rangle
\int_0^{nt}
a(n-s)^2 e^{im \tilde{\theta}_s} ds
+
O(n^{- \frac {\beta}{2}})
+
(martingale)
\eeq
Putting all together, we have 
\beq
I
&=&
\frac {mi}{2 \kappa} 
\cdot
\frac {-(2m+2)i}{2 \kappa} 
\cdot
\frac 12
\langle F g_{\kappa} \rangle
\int_0^{nt}
a(n-s)^2 e^{2mi \tilde{\theta}_s} ds
+
O(n^{- \frac {\beta}{2}}) + (martingale)
\eeq
By computing 
$II, III$
in a similar manner we obtain
\begin{eqnarray}
&&
e^{2mi \tilde{\theta}_{nt}}
\nonumber
\\
&=&
1 +
C_m
\int_0^{nt} 
a(n-s)^2 e^{2mi \tilde{\theta}_s} ds
+
O(n^{- \frac {\beta}{2}}) + (martingale)
\nonumber
\\
&=&
1 + C_m
\int_0^t n a(n-nu)^2 e^{2mi \tilde{\theta}_{nu}} du
+ O(n^{- \frac {\beta}{2}}) + (martingale)
\label{rho}
\end{eqnarray}
where 
\beq
C_m 
&:=&
\left(
\frac {m(2m+2)}{2(2 \kappa)^2}
\langle F g_{\kappa} \rangle
+
\frac {m(2m-2)}{2(2 \kappa)^2}
\langle F g_{- \kappa} \rangle
+
\frac {m^2}{\kappa^2} 
\langle F g \rangle
\right).
\eeq
Let 
$\sigma_F$
be the spectral measure of 
$L$
associated to 
$F$. 
Because
\beq
\langle F g_{\kappa} \rangle
=
\int_M F (L+2i \kappa)^{-1} F  dx
=
\int_{-\infty}^{0}
\frac {1}{\lambda + 2i\kappa}
d \sigma_F(\lambda), 
\eeq
$
Re \langle F g_{\kappa} \rangle
= 
Re \langle F g_{-\kappa} \rangle < 0$ 
and  
$
Re \langle F g \rangle < 0
$
so that 
\[
Re \; C_m < 0, 
\quad
m \ne 0.
\]
Take expectation on 
(\ref{rho})
and set
\beq
\rho_t^{(n)}
&:=&
{\bf E}[ e^{2mi \tilde{\theta}_{nt}} ] 
\\
f_n(t) 
&:=&
n a(n-nt)^2.
\eeq
Then 
\beq
\rho_t^{(n)}
&=&
1 + C_m \int_0^{t}
f_n(u) 
\rho_u^{(n)} du
+
g_n (u).
\eeq
where
\[
g_n (u) = O(n^{- \frac {\beta}{2}}), 
\quad
0 \le u \le 1 - n^{\beta-1}.
\]
It follows that
\beq
\rho^{(n)}_t
&=&
1 + C_m \int_0^t 
\left( f_n(s) + f_n(s) g_n(s) \right)
\exp \left(
C_m \int_s^t f_n(u) du
\right)
ds 
+ g_n(t)
\\
&=:&
A+B+C+D.
\eeq
The first two terms are equal to 
\beq
A+B
&=&
1 + 
C_m \int_0^t f_n(s)
\exp \left(
C_m \int_s^t f_n(u) du
\right)
ds
\\
&=&
1 + 
\left[
- 
\exp \left(
C_m \int_s^t f_n(u) du
\right)
\right]_0^t
\\
&=&
\exp \left(
C_m \int_0^t f_n(u) du
\right).
\eeq
Because
$Re \; C_m < 0$, 
we have
\[
\exp \left(
C_m \int_0^{1 - n^{\beta-1}}
f_n (u)du
\right)
=
\exp \left(
C_m
\int_0^{1 - n^{\beta-1}}
\frac {1}{1-s}
ds
\right)
\stackrel{n \to \infty}{\to}0.
\]
Similarly, 
$C, D = O(n^{- \frac {\beta}{2}})$. 
Therefore 
$\lim_{n \to \infty} \rho_t^{(n)} = 0$. 
\QED
\end{proof}
%

\subsection{Proof of Theorem \ref{Carousel}}
Take 
$0 < \epsilon < 1$, 
$0 < \beta < 1$
arbitrary, and let 
\[
l_1 = n(1 - \epsilon), 
\quad
l_2 = n - n^{\beta}
\]
so that 
$0 < l_1 < l_2$.
Moreover let 
$a_*$, $a^*=a_*+1 \in 2 \pi {\bf Z}$
satisfying 
\[
a_*  
\le 
\Psi_{1-\epsilon}^{(n)}(c)
< a^*.
\]
\begin{lemma}
\label{keylemma}
For
$n \gg 1$
and for 
$l_1 \le l \le l_2$
we have
\beq
\Psi_{l/n}^{(n)}(c) 
- 
\Psi_{1-\epsilon}^{(n)}(c)
&=&
2 \cdot \frac cn 
\left(
1 
+
O(n^{- \frac {\beta}{2}})
\right)
(l - l_1)
\\
&&+
\frac {1}{\kappa_0}
\int_{l_1}^{l} Re 
\left(
e^{2i \theta_s(\kappa_c)} - e^{2i \theta_s(\kappa_0)}
\right)
a(n-s) F(Y_s) ds.
\eeq
Thus by the comparison theorem, 
\beq
\Psi_{l/n}^{(n)}(c)
\ge
a_*, 
\quad
l_1 \le l \le l_2
\eeq
for sufficiently large 
$n$. 
\end{lemma}
\begin{proof}
By (\ref{theta-eq})
we have
\beq
&&
\Psi_{l/n}^{(n)} (c)
- \Psi_{1-\epsilon}^{(n)}(c)
=
2 \cdot \frac cn \cdot (l - l_1)
\\
&&\qquad+
\left(
\frac {2}{2 \kappa_c}
-
\frac {2}{2\kappa_0}
\right)
\int_{l_1}^{l} 
Re \left(
e^{2i \theta_s(\kappa_c)} - 1
\right)
a(n-s) F(Y_s) ds
\\
&&\qquad+
\frac {2}{2 \kappa_0}
\int_{l_1}^{l} Re 
\left(
e^{2i \theta_s(\kappa_c)} - e^{2i \theta_s(\kappa_0)}
\right)
a(n-s) F(Y_s) ds.
\eeq
Then the following estimate yields the conclusion. 
\beq
&&
\left(
\frac {2}{2 \kappa_c}
-
\frac {2}{2\kappa_0}
\right)
\int_{l_1}^{l} 
Re \left(
e^{2i \theta_s(\kappa_c)} - 1
\right)
a(n-s) F(Y_s) ds
\\
& \le &
\frac Cn
\int_{l_1}^{l} \frac {1}{\sqrt{n-s}} ds
\le
\frac Cn
\cdot
\frac {1}{n^{\beta/2}}
\cdot
(l - l_1).
\eeq
\QED
\end{proof}
The following lemma 
is an straightforward  consequence of 
(\ref{Decomposition}), 
Lemma \ref{J} and \ref{CDelta}.
\begin{lemma}
\label{Difference}
For 
$l_1 \le l \le l_2$, 
\beq
&&
\Psi_{l/n}^{(n)} (c)  - \Psi_{1-\epsilon}^{(n)}(c)
\\
&=&
2 \cdot \frac cn \cdot (l - l_1)
+
\frac {1}{\kappa_0}
Re
\left\{
\int_{l_1}^{l}
\left(
e^{2i \theta_s(\kappa_c)} - e^{2i \theta_s(\kappa_0)}
\right)
a(n-s) 
d M_s(\kappa_0)
\right\}
\\
&& \quad 
+ O(n^{\frac {\beta}{2}-1}) + O(n^{-\frac {\beta}{2}+\epsilon'}), 
\quad
0 < \epsilon' \ll 1.
\eeq
\end{lemma}
Let 
${\cal F}_t := \sigma
\left(
Y_s ; 0 \le s \le t 
\right)$.
\begin{proposition}
\label{Distance}
\beq
{\bf E}\left[
| \Psi_{1 - n^{\beta-1}}^{(n)}(c) - \Psi_{1-\epsilon}^{(n)}(c) | 
\, | \, {\cal F}_{l_1} 
\right]
& \le &
C
\left(
d(\Psi_{1 - \epsilon}^{(n)}(c), 2 \pi {\bf Z})
+
\epsilon
\right).
\eeq
\end{proposition}
This proposition 
follows from 
Lemma \ref{keylemma2}
below. 
\begin{lemma}
\label{keylemma2}
\beq
&(1)&\quad
{\bf E}[ | \Psi_{1 - n^{\beta-1}}^{(n)}(c) - a_* | | {\cal F}_{l_1}]
\le 
C
\left(
\epsilon
+
(\Psi_{1-\epsilon}^{(n)}(c) - a_*)
\right)
\\
&(2)& \quad
{\bf E}[ | \Psi_{1 - n^{\beta-1}}^{(n)}(c) - a^* | | {\cal F}_{l_1}]
\le 
C
\left(
\epsilon + (a^* - \Psi_{1-\epsilon}^{(n)}(c))
\right).
\eeq
\end{lemma}
\begin{proof}
(1)
By
Lemma \ref{Difference}
\begin{equation}
\left| 
{\bf E}[ \Psi_{1 - n^{\beta-1}}^{(n)}(c) - \Psi_{1-\epsilon}^{(n)}(c)
| {\cal F}_{l_1} ] 
\right|
\le 
2 \cdot \frac cn (l_2 - l_1)
+ o(1).
\label{two}
\end{equation}
By
Lemma \ref{keylemma}
\beq
&&
{\bf E}[ | \Psi_{1 - n^{\beta-1}}^{(n)}(c) - a_* | | {\cal F}_{l_1}]
=
\left|
{\bf E}[  \Psi_{1 - n^{\beta-1}}^{(n)}(c) - a_*  | {\cal F}_{l_1}]
\right|
\\
& \le &
\left| 
{\bf E}[ \Psi_{1 - n^{\beta-1}}^{(n)}(c) - \Psi_{1-\epsilon}^{(n)}(c) 
| {\cal F}_{l_1}] 
\right|
+
(\Psi_{1-\epsilon}^{(n)}(c) - a_*).
\eeq
Substituting 
(\ref{two}) 
and using 
$\frac cn (l_2 - l_1) \le C \epsilon$, 
we have the conclusion. 
\\
(2)
Letting 
\[
T := \inf \left\{
t \ge l_1 \, | \, 
\Psi_{t/n}^{(n)}(c) - a^* \ge 0 
\right\},
\]
we have
\beq
{\bf E}[ (\Psi_{1 - n^{\beta-1}}^{(n)}(c) - a^* )^+ 
| {\cal F}_{l_1}]
&=&
{\bf E} \left[
1(T \le l_2)
{\bf E}[ \Psi_{1 - n^{\beta-1}}^{(n)}(c) - a^* \, | \, {\cal F}_T ]
| {\cal F}_{l_1}
\right].
\eeq
If 
$T \le l_2$, 
then
$\Psi_{T/n}^{(n)}(c) = a^*$
so that by
Lemma
\ref{keylemma}, \ref{Difference}, 
\[
0 \le 
{\bf E}\left[
1(T \le l_2)
{\bf E}\left[
\Psi_{1 - n^{\beta-1}}^{(n)}(c) - 
\Psi_{T/n}^{(n)}(c) | {\cal F}_T 
\right]
| {\cal F}_{l_1}
\right]
\le
\frac cn 
(l_2 - l_1) + o(1).
\]
Therefore
\beq
{\bf E}[ (\Psi_{1 - n^{\beta-1}}^{(n)}(c) - a^* )^+ 
| {\cal F}_{l_1}]
& \le &
\frac cn \cdot (l_2 - l_1) + o(1)
\le 
C\epsilon.
\eeq
On the other hand
\beq
\left|
{\bf E} [ \Psi_{1 - n^{\beta-1}}^{(n)}(c) - a^* 
| {\cal F}_{l_1}] 
\right|
& \le &
\left|
{\bf E}[ \Psi_{1 - n^{\beta-1}}^{(n)}(c) - \Psi_{1-\epsilon}^{(n)}(c) 
| {\cal F}_{l_1}] 
\right|
+
(a^* - \Psi_{1-\epsilon}^{(n)}(c))
\\
& \le &
2 \cdot \frac cn \cdot (l_2 - l_1)
+
\left(
a^* - \Psi_{1-\epsilon}^{(n)}(c)
\right).
\eeq
Hence by using 
$|a| = - a + 2 a^+$
we have
\beq
&&
{\bf E}[ | \Psi_{1 - n^{\beta-1}}^{(n)}(c) - a^* | 
| {\cal F}_{l_1}] 
\\
&=&
\left|
{\bf E} \left[
(-\Psi_{1 - n^{\beta-1}}^{(n)}(c) + a^*) + 2 (\Psi_{1 - n^{\beta-1}}^{(n)}(c) - a^*)^+
| {\cal F}_{l_1}
\right]
\right|
\\
& \le &
C\epsilon
+
(a^* - \Psi_{1-\epsilon}^{(n)}(c)).
\eeq
\QED
\end{proof}
Let 
$\kappa_{\lambda} := \kappa_0 + \frac {\lambda}{n}$
and let 
$\theta^*_{n^{\beta}}(\kappa_{\lambda})$
be the solution to the following equation. 
\beq
\theta^*_{n^{\beta}}(\kappa_{\lambda})
:=
\kappa_{\lambda}(n^{\beta}-n)
+
\frac {1}{2 \kappa_{\lambda}}
\int_n^{n^{\beta}}
Re 
\left(
e^{2i \theta_s(\kappa_{\lambda})} - 1
\right)
a(n-s) F(Y_s) ds
\eeq
That is, 
$\theta^*_{n^{\beta}}(\kappa_{\lambda})$
is the Pr\"ufer phase function solved from the right endpoint. 
By Sturm-Liouville theory 
\begin{eqnarray}
&&
\sharp \{ \mbox{atoms of $\xi_n$ in }
[\lambda_1, \lambda_2]
\}
\nonumber
\\
&=&
\sharp \left(
\left[
2\theta_{n - n^{\beta}}(\kappa_{\lambda_1}) 
-
2\theta^*_{n^{\beta}}(\kappa_{\lambda_1}), 
2\theta_{n - n^{\beta}}(\kappa_{\lambda_2}) 
-
2\theta^*_{n^{\beta}}(\kappa_{\lambda_2})
\right]
\cap
2 \pi {\bf Z}
\right).
\label{SL}
\end{eqnarray}
\begin{lemma}
\label{Endpoint}
\beq
&&
\theta^*_{n^{\beta}}(\kappa_{\lambda})
-
\theta^*_{n^{\beta}}(\kappa_0)
\stackrel{P}{\to} 0
\eeq
\end{lemma}
\begin{proof}
By 
\cite{KN}, Lemma 6.4
we have
\[
{\bf E}\left[|
\theta^*_{t}(\kappa_{\lambda})
-
\theta^*_{t}(\kappa_0)
|
\right]
\le
C\cdot\frac tn + \frac {1}{\sqrt{n}}.
\]
Setting 
$t = n^{\beta}$
yields the conclusion. 
\QED
\end{proof}
{\it Proof of Theorem \ref{Carousel} }\\
Our goal is to show 
\beq
\sharp \left(
\mbox{atoms of $\xi_n$ in }
[0, \lambda_1], \cdots, [0, \lambda_d]
\right)
\to
\frac {1}{2\pi}
\left(
\Psi_{1-}(\lambda_1), 
\cdots, 
\Psi_{1-}(\lambda_d)
\right).
\eeq
We show this convergence for 
$d = 1$, 
for the general case follow similarly. 
By
Theorem \ref{SDE}
and 
Proposition \ref{Distance}, 
\begin{equation}
\Psi_{1 - n^{\beta-1}}^{(n)}(\kappa_{\lambda})
\stackrel{}{\to}
\Psi_{1-}(\lambda).
\label{sharp}
\end{equation}
for some subsequence. 
Letting 
$\lambda_1 = 0$, $\lambda_2 = \lambda$
in (\ref{SL}) 
we have
\beq
&&
\sharp \left(
\mbox{atoms of $\xi_n$ in }
[0, \lambda]
\right)
\\
&=&
\sharp \left(
\left[
2\theta_{n - n^{\beta}}(\kappa_0) 
-
2\theta^*_{n^{\beta}}(\kappa_0), 
2\theta_{n - n^{\beta}}(\kappa_{\lambda}) 
-
2\theta^*_{n^{\beta}}(\kappa_{\lambda})
\right]
\cap
2 \pi {\bf Z}
\right).
\eeq
The length of this interval is equal to, by 
Lemma \ref{Endpoint}, (\ref{sharp}), 
\beq
&&
2\left(
\theta_{n - n^{\beta}}(\kappa_{\lambda}) 
-
\theta^*_{n^{\beta}}(\kappa_{\lambda})
\right)
-
2\left(
\theta_{n - n^{\beta}}(\kappa_0) 
-
\theta^*_{n^{\beta}}(\kappa_0)
\right)
\\
& = &
\Psi_{1 - n^{\beta-1}}(\lambda) + o(1)_P
\stackrel{}{\to} 
\Psi_{1-}(\lambda).
\eeq
By conditioning on 
$Y_{n - n^{\beta}}$, 
we see that 
$\theta_{n - n^{\beta}}(\kappa_0)$
and 
$\theta^*_{n^{\beta}}(\kappa_0)$
are independent.
Thus by 
Theorem \ref{uniform} 
the left endpoint of this interval satisfies that its projection 
$\left\{
2\theta_{n - n^{\beta}}(\kappa_0) 
-
2\theta^*_{n^{\beta}}(\kappa_0)
\right\}_{2 \pi {\bf Z}}$
to 
$[0, 2 \pi)$  
converges to the uniform distribion on 
$[0, 2 \pi)$. 
Therefore
\beq
\sharp \left(
\mbox{atoms of $\xi_n$ in }
[0, \lambda]
\right)
\to
\Psi_{1-}(\lambda)
\eeq
proving Theorem \ref{Carousel}. 
\QED
%
\section{Appendix}
In this section 
we recall basic tools used in this paper. 
The content below are borrowed from 
\cite{KN}.
For 
$f \in C^{\infty}(M)$
let 
$R_{\beta} f 
:=(L + i\beta)^{-1}f$
$(\beta > 0)$,  
$R f :=L^{-1}(f - \langle f \rangle)$.
Then by 
Ito's formula, 
\beq
\int_0^t e^{i \beta s} f(X_s) ds
&=&
\left[
e^{i \beta s} (R_{\beta}f)(X_s) 
\right]_0^t
+
\int_0^t 
e^{i \beta s} d M_s(f, \beta)
\\
\int_0^t f(X_s) ds
&=&
\langle f \rangle t
+
\left[ (R f) (X_s) \right]_0^t 
+
M_t(f,0).
\eeq
$M_s(f, \beta), M_s(f, 0)$
are the complex martingales whose variational process satisfy
\beq
\langle M(f, \beta), M(f, \beta) \rangle_t
&=&
\int_0^t [ R_{\beta} f, R_{\beta} f] (X_s) ds, 
\\
\langle M(f, \beta), \overline{M(f, \beta)} \rangle_t
&=&
\int_0^t [ R_{\beta} f, \overline{R_{\beta} f}] (X_s) ds
\\
\langle M(f, 0), M(f, 0) \rangle_t
&=&
\int_0^t [ R f, R f] (X_s) ds, 
\\
\langle M(f, 0), \overline{M(f, 0)} \rangle_t
&=&
\int_0^t [ R f, \overline{R f}] (X_s) ds
\eeq
where
\beq
[f_1, f_2](x)
&:=&
L(f_1 f_2)(x)
-
(Lf_1)(x) f_2(x)
-
f_1(x) (Lf_2)(x)
\\
&=&
(\nabla f_1, \nabla f_2) (x).
\eeq
Then 
the integration by parts gives us the following formulas to be used frequently. 
\begin{lemma}
\label{PartialIntegration}
\beq
(1)
&&
\int_0^t 
b(s) e^{i \beta  s}
e^{i \gamma \tilde{\theta}_s}
f(X_s) ds
\\
&=&
\left[
b(s) e^{i \gamma \tilde{\theta}_s}
e^{i \beta  s}
(R_{\beta }f)(X_s)
\right]_0^t
- 
\int_0^t 
b'(s) e^{i \gamma \tilde{\theta}_s}
e^{i \beta  s}
(R_{\beta }f)(X_s)ds
\\
&& - 
\frac {i \gamma}{2 \kappa}
\int_0^t 
b(s) a(s) Re (e^{2i \theta_s}-1) 
e^{i \gamma \tilde{\theta}_s}
e^{i \beta  s}
F(X_s) (R_{\beta }f) (X_s) ds
\\
&& +
\int_0^t b(s) 
e^{i \beta  s} 
e^{i \gamma \tilde{\theta}_s}
d M_s(f, \beta). 
\eeq
\beq
(2) &&
\int_0^t b(s) e^{i \gamma \tilde{\theta}_s} f(X_s) ds
\\
&=&
\langle f \rangle 
\int_0^t b(s) e^{i \gamma \tilde{\theta}_s} ds
\\
&& + \left[
b(s) e^{i \gamma \tilde{\theta}_s} 
(Rf)(X_s) 
\right]_0^t 
 - \int_0^t
b'(s) e^{i \gamma \tilde{\theta}_s}
(Rf)(X_s) ds
\\
&& - \frac {i \gamma}{2 \kappa}
\int_0^t a(s)b(s) Re (e^{2i \theta_s} - 1)
e^{i \gamma \tilde{\theta}_s}
F(X_s) (Rf)(X_s) ds
\\
&& +
\int_0^t b(s) e^{i \gamma \tilde{\theta}_s}
dM_s(f, 0).
\eeq
\end{lemma}
We will also use following notation for simplicity. 
\beq
g_{\kappa} &:=& (L+2i \kappa)^{-1}F, 
\quad
g := L^{-1}(F - \langle F \rangle),
\\
h_{\kappa, \beta} &:=& (L + 2i \beta \kappa)^{-1}F g_{\kappa}
\\
M_s(\kappa) &:=& M_s (F, 2 \kappa), 
\quad
M_s := M_s (F,  0),
\\
\widetilde{M}_s^{(\beta)}(\kappa)
&:=&
M_s (F g_{\kappa}, \beta \kappa), 
\quad
\widetilde{M}_s
:=
M_s (F g_{\kappa}, 0).
\eeq
\begin{lemma}
\label{Inverse}
Let 
$\Psi_n, n=1,2, \cdots$, 
and 
$\Psi$
are continuous and increasing functions defined on a open set 
$K \subset {\bf R}$ 
such that
$\lim_{n \to \infty}\Psi_n(x)=\Psi(x)$ 
pointwise. 
If 
$y_n \in Ran \;\Psi_n$, 
$y \in Ran \;\Psi$
and
$y_n \to y$, 
then it holds that 
\[
\Psi_n^{-1}(y_n) 
\stackrel{n \to \infty}{\to}
\Psi^{-1}(y).
\]
\end{lemma}

\vspace*{1em}
\noindent {\bf Acknowledgement }
This work is partially supported by 
JSPS grant Kiban-C no.22540140.


\end{document}